\documentclass[12pt,a4paper]{article}
\usepackage[utf8]{inputenc}
\usepackage{istgame}
\usepackage{tikz}
\usepackage{xparse}
\usepackage{expl3}
\usepackage{amsmath,amssymb, amsfonts,amsthm}
\usepackage{bm}
\usepackage{mathtools}
\usepackage{enumerate}
\usepackage{enumitem}
\usepackage[round,sort]{natbib}
\usepackage{ulem}
\usepackage{appendix}
\usepackage{chngcntr}
\usepackage{apptools}
\AtAppendix{\counterwithin{lemma}{subsection}}
\AtAppendix{\counterwithin{theorem}{subsection}}
\AtAppendix{\counterwithin{remark}{subsection}}
\AtAppendix{\counterwithin{example}{subsection}}
\AtAppendix{\counterwithin{proposition}{subsection}}
\AtAppendix{\counterwithin{definition}{subsection}}
\AtAppendix{\counterwithin{corollary}{subsection}}

\usepackage{geometry}
\usepackage{hyperref}
\usepackage{cleveref}
\usepackage[misc]{ifsym}
\usepackage[singlespacing]{setspace}
\usepackage[e]{esvect}

\hypersetup{
colorlinks = true,
urlcolor = green,
linkcolor = red,
citecolor = blue}
\newlength\mylen
\newlist{mycases}{enumerate}{1}
\setlist[mycases,1]{label=\textbf{Case~\arabic*.}, 
  labelwidth=\dimexpr-\mylen-\labelsep\relax,leftmargin=0pt,align=right}

\newtheorem{theorem}{Theorem}
\newtheorem{lemma}{Lemma}
\newtheorem{corollary}{Corollary}
\newtheorem{proposition}{Proposition}
\newtheorem{definition}{Definition}

\theoremstyle{remark}
\newtheorem{remark}{Remark}
\newtheorem{example}{Example}

\newcommand\eqdist{\stackrel{\mathclap{\mbox{d}}}{=}}
\DeclareMathOperator*{\esssup}{ess\,sup}
\DeclareMathOperator*{\essinf}{ess\,inf}

\DeclareMathOperator*{\argmin}{arg\,min}
\newcommand{\defin}[1]{{\bfseries\upshape #1}}
\renewcommand{\le}{\leqslant}
\renewcommand{\ge}{\geqslant}
\renewcommand{\leq}{\leqslant}

\newcommand{\Prob}{\mathbf{P}}
\newcommand{\ProbQ}{\mathbf{Q}}
\newcommand{\E}{\operatorname{\mathbf{E}}}
\newcommand{\V}{\operatorname{\mathbf{V}}}

\newcommand{\VaR}{\operatorname{VaR}}
\newcommand{\AVaR}{\operatorname{AVaR}}

\newcommand{\SVaR}{\operatorname{S-VaR}}
\newcommand{\SAVaR}{\operatorname{S-AVaR}}

\newcommand{\Cov}{\operatorname{Cov}}

\mathchardef\mhyphen="2D 

\newcommand*{\defcol}{\mathrel{\vcenter{\baselineskip0.5ex \lineskiplimit0pt
\hbox{\scriptsize.}\hbox{\scriptsize.}}}%
=}

\title{The limitations of comonotonic additive risk measures: a literature review}
\author{Samuel S. Santos\thanks{\textbf{Corresponding author}. \Letter \hspace{0.1cm} \textcolor{blue}{ssolgons@uwaterloo.ca}. 200 University Avenue West, Waterloo, ON, Canada N2L 3G1. The author thanks for the support from the Coordination for the Improvement of Higher Education Personnel (CAPES) under grant number 88882.439088/2019-01. Declarations of interest: none.} \\ \small{University of Waterloo} \and Marcelo Brutti Righi\thanks{\hspace{0.1cm} \Letter \hspace{0.1cm} marcelo.righi@ufrgs.br. The author thanks for the support from the Brazilian National Council for Scientific and Technological Development (CNPq) under grant number 302614/2021-4. Declarations of interest: none.}\\ \small{Federal University of Rio Grande do Sul} \and Eduardo Horta \thanks{\hspace{0.1cm} \Letter \hspace{0.1cm} eduardo.horta@ufrgs.br. The author thanks for the support from the Brazilian National Council for Scientific and Technological Development (CNPq) under grant number 438642/2018-0. Declarations of interest: none.}\\ \small{Federal University of Rio Grande do Sul}}

\date{\today}

\begin{document}
\maketitle

\begin{abstract}
Risk measures satisfying the axiom of comonotonic additivity are extensively studied, arguably because of the plethora of results indicating interesting aspects of such risk measures. Recent research, however, has shown that this axiom is incompatible with properties that are central in specific contexts. In this paper, we present a literature review of these incompatibilities. In addition, we use the Choquet representation of comonotonic additive risk measures to show they cannot be surplus invariant.
\end{abstract}

\section{Introduction}\label{introduction1}

Risk measures are a cornerstone of modern financial practice: institutions and investors use them to quantify potential losses and variability; regulators, to determine the regulatory capital that institutions such as banks, investment funds, and insurance companies, must hold as a buffer against their potential losses.
Theoretical research on risk measures met its current, axiomatic tone in \citet*{artzner1999}, followed by \cite{delbaen2002coherent}, \cite{follmer2002convex}, \cite{frittelli2002putting}, \cite{kusuoka}, and \cite{acerbi2002spectral}), to name a few. According to this paradigm, a risk measure is a functional $\rho$ that assigns to each financial position, represented by a random variable $X$, a certain real number $\rho(X)$, interpreted as the financial risk of $X$. Here, we follow the mainstream literature and interpret $\rho(X)$ as the amount of regulatory capital that must be held by financial institutions whose position is $X$, to cover potential losses from $X$.

One of the great merits of this literature on risk measures has been its ability to translate intuitive aspects of the notion of risk into well formalized mathematical axioms. In the present paper, we focus on measures of financial risk satisfying the axiom of comonotonic additivity, that is, risk measures $\rho$ for which the identity $\rho(X_{1}+X_{2})=\rho(X_{1})+\rho(X_{2})$ holds whenever the assets $X_{1}$ and $X_{2}$ are increasing functions of a common underlying asset. This axiom occupies a distinguished place in the theory (see, for instance, \cite{kusuoka}, \cite{acerbi2002spectral}, \cite{dhaene2003economic}, \cite{dhaene2004capital}, \cite{deelstra2011overview}, \cite{ekeland2012}, \cite{kou2016}, \cite{rieger2017characterization}, \cite{koch2018}, \cite{wang2020characterization}). Such imminence comes from the strong intuition that comonotonicity conveys: first, two comonotonic random variables always move in the same direction, so they do not hedge each other. As a reasonable extension of this fact, one could say---and this is indeed the traditional view--- that ``there is no benefit (and no downside either) in pooling comonotonic random variables together". For risk measures, this statement translates into the claim that ``for mutually unhedgeable assets, the risk of the sum should be the sum of the risks''.

The present paper offers a comprehensive literature review reporting properties that are necessarily absent for key classes of comonotonic additive risk measures.
Such incompatibilities may take two forms: first, desirable features that are absent in most comonotonic additive risk measure (sections, \ref{lossesonly.1}, \ref{eligible.1}, \ref{elicit}, and \ref{section_dynamics}); second, features that are idiosyncratic and (possibly) troublesome but are possessed by all reasonable comonotonic additive risk measures (sections \ref{port} and \ref{risk.aversion}).

\subsection{Roadmap}\label{roadmap}
\hypersetup{linkcolor=black}
In the \hyperref[s:appendix]{Appendix}\hypersetup{linkcolor=red} we present the elementary theory on comonotonic random variables and risk measures, focusing on comonotonic additive risk measures and their integral representations. The \textit{connoisseur} will find no novelty there, but may want to give it a quick overview to get familiarized with the definitions we adopt. For the readers skipping the appendix, just keep in mind that we work with an atomless probability space $(\Omega,\mathcal{F},\Prob)$, and that the net present value of the financial positions are represented, for convenience, by random variables in $\mathcal{X}\coloneqq L^{\infty}(\Omega,\mathcal{F},\Prob)$. Also, we use the term \defin{risk measure} to refer to any functional $\rho\colon \mathcal{X}\rightarrow \mathbb{R}$, and the term \defin{acceptance set} to refer to any non-empty set $\mathcal{A}\subsetneq \mathcal{X}$.
As usual, given a risk measure $\rho$ we write $\mathcal{A}_\rho\coloneqq \{X\in\mathcal{X}\colon\ \rho(X)\leq0\}$ and given an acceptance set $\mathcal{A}$ we put $\rho_{\mathcal{A}}(X)\coloneqq \inf\{m\in\mathbb{R}\colon\ X+m\in\mathcal{A}\}$, $X\in\mathcal{X}$.
Exceptions to these definitions are explicitly mentioned. Most results presented in this paper are not ours, but collected from the literature. In these cases, the statements always begin with the respective citation. Any formal statement not beginning with a citation is new.

We begin our bibliographical review in \Cref{lossesonly.1}, discussing the difficulties of using comonotonic additive risk measures to determine regulatory capital in the context of limited liability. The view presented by the papers in this topic is that, when financial firms have limited liability, the regulatory capital should be determined solely by the potential losses incurred by the financial institutions, being therefore insensitive to the size and probability of surpluses \citep{staum2013, cont2013loss, koch2015, koch2017, he2018}. This insensitivity can be captured by different axioms that reflect the notion of surplus invariance. In \Cref{lossesonly.1}, we show that comonotonic additive risk measures satisfying surplus invariance (as defined in \cite{koch2017}, \cite{he2018}, and \cite{gao2020surplus}) are necessarily quantiles. Our proofs are based on the Choquet representation, contrasting with \cite{he2018}, who used acceptance sets to obtain (more general) results regarding positive homogeneous risk measures. Also based on the Choquet representation, we show that the maximum loss is the only convex comonotonic additive risk measure satisfying the property of surplus invariance.

In \Cref{eligible.1}, we discuss the findings of \cite{koch2018} showing that the regulator must choose between using a comonotonic additive risk measure or allowing banks to use risky assets to compose their regulatory reserves (which should be allowed according to the Basel Committee's guidelines \citep{bcbs2019feb}). \cite{koch2018} showed (under mild conditions) that the cost of insisting on both properties is prohibitive, as the regulator would be forced to accept arbitrarily large, fully-leveraged positions. As a consequence, when risky eligible assets are considered, spectral risk measures lose their comonotonic additivity (even the value-at-risk and the average value-at-risk).

In \Cref{elicit}, we discuss the lack of elicitability of comonotonic additive risk measures. The property of elicitability for risk measures lies at the heart of the recent literature studying the basic principles that allow one to compare different risk forecasting procedures. Therefore, the lack of elicitability imposes additional difficulties in comparing different risk forecasting procedures. Essentially, a risk measure is elicitable if it minimizes a specific expected score, which allows us to rank risk forecasting procedures in a meaningful manner. As noticed in \cite{gneiting2011}, the lack of elicitability or the usage of inadequate score functions can lead to misleading evaluations of the forecasting procedures' relative performances. Elicitable risk measures, however, are scarce \citep{weber2006,gneiting2011,bellini2015_OnElicitableRiskMeasures, ziegel2016}. In particular, \cite{ziegel2016} showed that, if a coherent risk measure is elicitable and law invariant, then it is an expectile. Also, if in addition the risk measure is comonotonic additive, then it corresponds to the expected loss w.r.t.\ the physical probability \citep{ziegel2016}. For the non-coherent case, \cite{kou2016} showed that there is no law invariant monetary risk measure that is comonotonic additive and elicitable, except the expected loss and the value-at-risk. Taken together, these results tell us that, if the property of elicitability is of utmost importance and coherency is desirable, then the cost of requiring a risk measure to be comonotonic additive is that we would end up confined to the expected loss, which is inadequate to measure tail risk; if coherence is dropped, we can still employ the value-at-risk.

All the properties and papers mentioned so far were developed in a one-period framework. Despite being remarkably useful, this framework's potential to measure risk in a dynamic setting is limited, as it does not allow the risk to depend on new information. In \Cref{section_dynamics}, we discuss comonotonic risk measures in the dynamic framework, where the risk of a position is measured at each period. A central topic in this context is the property of time-consistency, which defines how the risk at different periods should relate. Loosely speaking, this property requires that, if the risk of $X$ is greater than that of $Y$ at the period $t+1$ with probability one, then the same must hold at $t$.
We discuss two incompatibilities between time-consistency and the comonotonic additivity axiom. First, \cite{kupper2009} showed that the the entropic is the unique monetary dynamic risk measure that is law invariant, relevant, and time-consistent. Since the static entropic risk measures are not comonotonic additive, the findings of \cite{kupper2009} imply that one cannot construct a time-consistent dynamic risk measure with comonotonic additive components. Second, we present the conflict between time-consistency and comonotonic additivity discovered by \cite{delbaen2021commonotonicity}. As he showed, the unique (not necessarily law invariant) risk measure that is coherent, time-consistent, and comonotonic additive is an expected loss with respect to an absolutely continuous probability measure.

A major appeal of comonotonic additive risk measures is their spectral representations, which are intuitive and allow us to explicitly manipulate general risk measures in this class (see \cite{kusuoka}, \cite{acerbi2002spectral}, \cite{follmer2016stochastic}, and \cite{wang2020characterization} for details). Because of these representations, one could expect them to be valuable tools in applied problems, in particular, in portfolio selection problems. As shown in \cite{brandtner2013conditional}, however, such application of spectral risk measures leads to two (possible problematic) idiosyncrasies that we discuss in \Cref{port}. First, recall that in the mean-variance framework, the optimal weights of a portfolio can be found by maximizing expected returns subject to a certain level of variance or, equivalently, by minimizing the variance subject to a certain level of expected return. As a first quirk in the context of portfolio optimization, \cite{brandtner2013conditional} showed that those problems are no longer equivalent when the variance is replaced by a spectral risk measure. A second complication emerging in such problems---also brought about by \cite{brandtner2013conditional}---is that when a spectral risk measure is used, the solution tends to be at the corners. As a consequence, when short sales are allowed, the investors either invest an infinite amount in the tangency portfolio (by short selling the risk-free asset) or invest zero in the risky portfolio. If short sales are restricted, the investor invests all or nothing of her money in the risk-free asset.

The usage of comonotonic additive functionals is far from being restricted to the field of risk measures. In fact, these functionals' early roots lie in non-expected utility theory \citep{yaari1987dual,schmeidler1989subjective}, where the notion of risk aversion of \cite{arrow1965aspects} and \cite{pratt1964risk} has fundamental importance. In classical utility theory, we can compare the risk aversion of two agents through their certainty equivalents or, equivalently, through their Arrow-Pratt coefficients of risk aversion (and these comparisons coincide). In \Cref{risk.aversion}, we discuss the findings of \cite{brandtner2015decision} showing that, if the agents' preferences are represented by a spectral risk measure, their relative risk aversion may be different depending on if it is measured through the certainty equivalents or the Arrow-Pratt coefficient. As argued in \cite{brandtner2015decision}, this lack of consistency makes the usage of the Arrow-Pratt coefficient troublesome, because the relative risk aversion of two individuals can be different when, instead, it is measured by the certainty equivalent. \cite{brandtner2015decision} extended their analysis to the framework of \cite{ross1981some}, which considers a random level of initial wealth. As an extension of the previously mentioned inconsistency, they showed that, for agents whose preferences are represented by spectral risk measures, the ordering based on the Arrow-Pratt coefficient does not necessarily coincide with the ordering based on the coefficient of risk aversion of \cite{ross1981some}. We summarize and conclude the paper in \Cref{conclusion.1}.

\section{Excess Invariance}\label{lossesonly.1}
As argued in \cite{staum2013}, \cite{koch2015}, \cite{koch2017}, and \cite{he2018}, the social justification of a regulator should be to secure the bank's creditors against the risk of default. According to this view, the risk of a financial institution should be determined based exclusively on the negative part of the position it holds.
\begin{definition}\label{def.staum}
\citep{staum2013,cont2013loss,gao2020surplus} A risk measure $\rho\colon \mathcal{X}\rightarrow\mathbb{R}$ is \defin{excess invariant} if $\rho(X)=\rho(-X^{-})$ for all $X \in \mathcal{X}$.\footnote{We adopt the concise notation $-X^{-}=\min\{X,0\}$.}
\end{definition}
The property of excess invariance also fits the needs of financial clearing facilities, that make the custody of the money of multiple market participants and must guarantee that all due payments are made. These organization must manage the risk of default from the parties in a contract, mainly by defining the margin requirements that investors must deposit to ensure they can face possible losses. The size of the investors gains, however, do not influence their risk of default and, therefore, the property of excess invariance is especially suited for managing risk of default in this context.

\begin{remark}
Excess invariant risk measures---let us temporally denote them as $\Tilde{\rho}$---can always be constructed from traditional risk measures, say $\rho$, through $\Tilde{\rho}(X)\coloneqq\rho(-X^{-})$, which is real-valued for all $X \in \mathcal{X}$. 
Risk measures in the form $\Tilde{\rho}$ induce acceptance sets according to $\mathcal{A}(\beta)\coloneqq\{X \in \mathcal{X}\colon \Tilde{\rho}(X)\le \beta\}$ for some $\beta >0$, which represents a level of risk tolerance. This interpretation fits into the traditional framework, although if $\Tilde{\rho}$ is excess invariant, the surplus of the position $X \in \mathcal{X}$ does not play a role in determining if it is tolerable or not. Also, notice that when $\Tilde{\rho}$ is non-negative, using $\beta=0$ as in the traditional framework could be excessively restrictive.
\end{remark}
As shown in \cite{staum2013} (Proposition 3.2) and \cite{gao2020surplus} (Proposition 5), there is a direct conflict between the properties of excess invariance and cash additivity. Since comonotonic additivity implies $\rho(X+c)=\rho(X)+c\rho(1)$, $\forall X \in \mathcal{X}$ and $c \in \mathbb{R}$---which is a weaker form of cash additivity as defined in \Cref{def.rho.1}---there is also a tension between the properties of comonotonic additivity and excess invariance. 
\begin{proposition}\label{prop.staum.1}
Let $\rho$ be a monotone, excess invariant risk measure. Then the following holds:
\begin{enumerate}[noitemsep,nosep]
\item\label{prop.staum.1.1} If $\rho$ is normalized, then it is non-negative.
\item\label{prop.staum.1.2} If $\rho$ is non-zero, then it is not comonotonic additive.
\end{enumerate}
\end{proposition}
\begin{proof}
The first item follows directly from Proposition 3.1 of \cite{staum2013}. The proof of item \ref{prop.staum.1.2} goes by contradiction, so let us begin assuming that $\rho$ satisfies all the properties mentioned in the statement. By Proposition 2.5 of \cite{koch2018}, non-zero monotone and comonotonic additive risk measures are positive homogeneous and, therefore, normalized. Hence, item \ref{prop.staum.1.1} implies that $\rho$ is non-negative, and the non-zero property yields an $X \in \mathcal{X}$ with $\rho(X)>0$. Also, notice that $X+\Vert X\Vert_{\infty}\ge0$ and, therefore, $-(X+\Vert X\Vert_{\infty})^{-}=0$. Excess invariance implies $\rho(X+\Vert X\Vert_{\infty})=\rho(-(X+\Vert X\Vert_{\infty})^{-})$, so normalization tells us that $\rho(X+\Vert X\Vert_{\infty})=0$. On the other hand, comonotonic additivity and positive homogeneity imply that $\rho(X+\Vert X\Vert_{\infty})=\rho(X)+\Vert X\Vert_{\infty}\rho(1)$. But excess invariance gives $\rho(c)=\rho(0)$ for all $c\ge 0$, and normalization implies $\rho(0)=0$. Therefore, we conclude that $\rho(X+\Vert X\Vert_{\infty})=\rho(X)>0$, which is absurd.
\end{proof}
\begin{remark}
The contradiction in the above proof comes from the fact that i) the monotonicity and excess invariance of $\rho$ allow us to ``neutralize" the risk of $X$ by adding $\Vert X\Vert_{\infty}$ to obtain $\rho(X+\Vert X\Vert_{\infty})=0$ and ii) the comonotonic additivity property implies $\rho(X+\Vert X\Vert_{\infty})=\rho(X)+\Vert X\Vert_{\infty}\rho(1)=\rho(X)>0$.
\end{remark}

As the property of excess invariance is so restrictive, \cite{koch2017}, \cite{he2018}, and \cite{gao2020surplus} considered the following properties instead:
\begin{definition}
A risk measure $\rho:\mathcal{X}\rightarrow \mathbb{R}$ is \defin{surplus invariant subject to positivity} ($\text{SI}^+$) if $\rho(X)=\rho(-X^-)$ whenever $\rho(X)\ge 0$. An acceptance set $\mathcal{A}$ is \defin{surplus invariant} if whenever $X \in \mathcal{A}$ and $Y \in \mathcal{X}$ are such that $Y^{-}\le X^{-}$ we have that $Y \in \mathcal{A}$.
\end{definition}

The above definitions are mirrored images of each other: $SI^+$ risk measures induce surplus invariant acceptance sets, and surplus invariant acceptance sets induce $SI^+$ risk measures (for more details see Proposition 3.10 of \cite{koch2015} and Proposition 4 of \cite{gao2020surplus}). Also, is not difficult to show that excess invariance (of a risk measure $\rho$) is strictly stronger than $SI^+$ (for risk measures) and surplus invariance for the corresponding acceptance set, $\mathcal{A}_\rho$. As the following examples show, there exist whole families of convex cash additive risk measures satisfying $SI^+$.
\begin{example}
(\cite{follmer2002convex}) Let $l:\mathbb{R}^+ \rightarrow \mathbb{R}^+$ be a non-decreasing and non-constant function and let $c \ge 0$ belong to the interior of the convex hull of the range of $l$. The \defin{shortfall risk measure} induced by $l$ is defined as
\begin{equation}\label{example.shortfall}
\rho_{l}(X)=\inf\{m \in \mathbb{R}:\E[l([X+m]^-)]\le c\},\quad \forall X \in \mathcal{X}.
\end{equation}
Well-known specifications that one can take for the loss function are $l(x)=e^{\beta x}$ with $\beta >0$, and $l(x)=(1-\lambda)x^+ -\lambda x^-$. The first is associated with the entropic risk measure (see \Cref{teo.kupper}), and the second loss is associated with the expectiles (see \Cref{def.expectiles}). One can easily check that $\rho_{l}$ is cash additive, for any non-decreasing non-constant loss function. To see that any shortfall risk measure is surplus invariant, notice that, $X^-=Y^-$ implies $(X+m)^-=(Y+m)^-$ for all $m\ge 0$. In particular, this holds if we take $Y=-X^-$, which implies $(X+m)^-=(-X^- +m)^-$ for $m \ge 0$. If $\rho_l(X)\ge 0$, it holds that
\begin{equation}
\{m \in \mathbb{R}:\E[l([X+m]^-)]\le c\}=\{m \in \mathbb{R}:\E[l([-X^-+m]^-)]\le c\},
\end{equation}
and, therefore $\rho_l(X)=\rho_l(-X^-)$. To see that shortfall risk measures are not excess invariant, notice that, for a loss function $l$ and a constant positive random variable $X=x_0>\rho_l(0)$, it holds $\rho_{l}(X)=\rho_{l}(0)-x_0<\rho_{l}(-X^-)=\rho_{l}(0)$.
\end{example}

\begin{example}
The acceptance set associated with $\VaR_p$ is given by $\mathcal{A}=\{X \in \mathcal{X}:\Prob(X<0)\le p\}$. To see that it is surplus invariant, take $X \in \mathcal{A}$ and let $Y \in \mathcal{X}$ be such that $Y^{-}\le X^{-}$. Then we have
\begin{equation*}
   \Prob(Y<0)=\Prob(-Y^{-}<0)\le \Prob(-X^{-}<0)=\Prob(X<0)\le p.
\end{equation*}
This, in turn, implies that $Y \in \mathcal{A}$ and that $\mathcal{A}$ is surplus invariant. Since surplus invariant acceptance sets induce $SI^+$ risk measures, it holds that $\rho_{\mathcal{A}}=\VaR_p$ is $SI^+$. In fact, as the next proposition shows, the value at risk is the unique $SI^+$ risk measure in the class of monetary law-invariant comonotonic additive risk measures. Notice that we provide an alternative integral representation of these risk measures in \Cref{teo.kou}.
\end{example}

\begin{lemma}\label{lemma.var.si}
Let $X \in \mathcal{X}$. If $h \in \mathcal{H}^*$ is such that $h(t) \in \{0,1\}$ for all $t \in [0,1]$, then there exists $\alpha \in [0,1]$ such that $\rho_{h}(X)=\VaR_{\alpha}(X)$ or $\rho_{h}(X)=-q_{X}^{+}(\alpha)$.
\end{lemma}
\begin{proof}
As a preliminary, define $t_s:=\sup\{t \in [0,1):h(t)=0\}$, $t_i:= \inf\{t \in (0,1]: h(t)=1\}$, and notice that, since $h(t)\in \{0,1\}$ for all $t \in [0,1]$, it holds that $t_s=t_i=t^* \in [0,1]$. If $t^* =0$, then it holds that $h(t^*)=0$ and, because $h(t)\in \{0,1\}$ for all $t \in [0,1]$, it follows that $h(t)=1_{(t>0)}$. Therefore, 
\begin{equation}
h\circ \Prob(-X>x)=\begin{cases}
1, \mbox{if } \Prob(-X>x)>0\\
0, \mbox{if } \Prob(-X>x)=0,
\end{cases}
\end{equation}
which is equivalent to
\begin{equation}
h\circ \Prob(-X>x)=\begin{cases}
1, \mbox{if } x<\esssup(-X)\\
0, \mbox{if } x \ge \esssup(-X).
\end{cases}
\end{equation}
For $X \in \mathcal{X}$ such that $\esssup(-X) \ge 0$ ($\esssup(-X) < 0$, respectively), it holds that
\begin{equation}
\rho_h(X)=\int_{[0,\esssup(-X))}1dx\quad \left(\rho_h(X)=-\int_{[\esssup(-X),0]}dx,\; \text{respectively}\right)
\end{equation}
and, in both cases, it holds that $\rho_h(X)=\esssup(-X)=\VaR_{0}(X)$ by the definition presented in \Cref{example.var}.
For the case $t^*=1$, it holds $h(t)=1_{(t=1)}(t), \; \forall t \in [0,1]$. Therefore, it is true that
\begin{equation}
h\circ \Prob(-X>x)=\begin{cases}
1, \mbox{if } x<\essinf(-X)\\
0, \mbox{if } x > \essinf(-X).
\end{cases}
\end{equation}
As before, one can easily check that, for the distortion function above, it holds that $\rho_h(X)=\essinf(-X)$, which equals $\VaR_{1}(X)$.

For $t^* \in (0,1)$, we must consider two cases: first, $h(t^*)=0$ and, second $h(t^*)=1$. In the first case, $h(t)=1_{(t > t^*)}(t), \; \forall t \in [0,1]$. Therefore,
\begin{equation}
h\circ \Prob(-X>x)=\begin{cases}
1, \mbox{if } \Prob(-X>x)>t^*\\
0, \mbox{if } \Prob(-X>x) \le t^*,
\end{cases}
\end{equation}
which is equivalent to
\begin{equation}
h\circ \Prob(-X>x)=\begin{cases}
1, \mbox{if } x<q_{-X}(1-t^*)\\
0, \mbox{if } x \ge q_{-X}(1-t^*).
\end{cases}
\end{equation}
For $X \in \mathcal{X}$ such that $q_{-X}(1-t^*)\ge 0 $ ($q_{-X}(1-t^*) < 0$, respectively), it holds that
\begin{equation}
\rho_h(X)=\int_{[0,q_{-X}(1-t^*))}1dx\quad \left(\rho_h(X)=-\int_{[q_{-X}(1-t^*),0])}dx,\; \text{respectively}\right)
\end{equation}
and, in both cases, it holds that $\rho_h(X)=q_{-X}(1-t^*)=\VaR_{t^*}(X)$.
For the second case, where $h(t^*)=1$, it holds that $h(t)=1_{(t \ge t^*)}(t), \; \forall t \in [0,1]$. Therefore,
\begin{equation}
h\circ \Prob(-X>x)=\begin{cases}
1, \mbox{if } x<q_{-X}^{+}(1-t^*)\\
0, \mbox{if } x > q_{-X}^{+}(1-t^*).
\end{cases}
\end{equation}
For $X \in \mathcal{X}$ such that $q_{-X}^{+}(1-t^*) \ge 0$ ($q_{-X}^{+}(1-t^*)<0,\;\text{respectively}$), it holds that
\begin{equation}
\rho_h(X)=\int_{[0,q_{-X}^{+}(1-t^*))}1dx\quad \left(\rho_h(X)=-\int_{(q_{-X}^{+}(1-t^*),0])}dx,\; \text{respectively}\right)
\end{equation}
and, in both cases, it holds that $\rho_h(X)=q_{-X}^{+}(1-t^*)$, which concludes the proof.
\end{proof}

\begin{proposition}\label{SI+}
Let $\rho$ be monetary law-invariant and comonotonic additive risk measure. If $\rho$ is, in addition, surplus invariant subject to positivity, then there exists $\alpha \in [0,1]$ such that, for any $X \in \mathcal{X}$, it holds that $\rho(X)=\VaR_{\alpha}(X)$ or $\rho(X)=-q_{X}^+(\alpha)$.
\end{proposition}
\begin{proof}
For any monetary risk measure $\rho$ that is comonotonic additive and law-invariant, there exists $h \in \mathcal{H}^*$ such that $\rho = \rho_h$. By \Cref{lemma.var.si}, it suffices to show that, if $\rho_h$ is $SI^+$, then $h \in \mathcal{H}^*$ is such that $h(t) \in \{0,1\}$ for all $t \in [0,1]$. Our proof goes by contradiction. Assume that $h(t_0) \in (0,1)$ for some $t_0 \in (0,1)$ and let $a_1<0<a_2$ be constants such that
\begin{equation}
a_2h(t_0) \ge a_1[h(t_0)-1].
\end{equation}
Consider a random variable $X \in \mathcal{X}$ such that
\begin{equation}
\Prob(-X >x)=\begin{cases}
1,\mbox{if } x<a_1\\
t_0, \mbox{if } x \in [a_1,a_2)\\
0, \mbox{if }x \ge a_2.
\end{cases}
\end{equation}
Thus, we have
\begin{equation}
\Prob(-X >x)=\begin{cases}
1,\mbox{if } x<a_1\\
h(t_0), \mbox{if } x \in [a_1,a_2)\\
0, \mbox{if }x \ge a_2
\end{cases}
\end{equation}
and, therefore,
\begin{align*}
\rho_h(X)&=\int_{a_1}^{0}[h\circ \Prob(-X>x)-1]dx +\int_{0}^{a_2}h \circ \Prob(-X>x)dx\\
& = \int_{a_1}^{0}[h(t_0)-1]dx +\int_{0}^{a_2}h(t_0)dx\\
& = -a_1[h(t_0)-1]+a_2h(t_0) \ge 0.
\end{align*}
Notice that, for $x <0$, it holds that $\Prob(-X^+ >x)=\Prob(-X>x)$ and, for $x \ge 0,$ $\Prob(-X^+>x)=0$. Therefore,
\begin{equation}
\Prob(-X^+ > x)=\begin{cases}
1, \mbox{if } x<a_1\\
t_0, \mbox{if } x \in [a_1,0]\\
0, \mbox{if } x \ge 0.
\end{cases}
\end{equation}
Based on the above distribution, it holds that
\begin{equation}
\rho_h(X^+)=\int_{a_1}^0 [h(t_0)-1]dx=-a_1[h(t_0)-1]<0
\end{equation}
Now, since $X=X^+ -X^-$ and the random variables $X^+$ and $-X^-$ are comonotonic, it holds that $\rho_h(X)=\rho_h(X^+)+\rho_h(-X^-), \; \forall X \in \mathcal{X}$. Since $\rho_h(X^+)<0$, it holds that $\rho_h(X)<\rho_h(-X^-)$, from which we conclude that $\rho_h$ is not $SI^+$.
\end{proof}
The above proof is based on the Choquet integral representation of comonotonic additive risk measures provided (\Cref{teo.kou}). We showed that, whenever a comonotonic additive risk measure is defined by more than a single quantile (say two quantiles, $\alpha_1$ and $\alpha_2$), one can always find a financial position such that $q_X(\alpha_1)<0<q_X(\alpha_2)$ with $|\rho_h(-X^-)|<\rho_h(X^+)$. In the following result, we apply this reasoning to the Kusuoka representation of comonotonic additive coherent risk measures (\Cref{kusuoka_com.1}). As shown below, the maximum loss ($ML(X)=\esssup(-X),\; \forall X \in \mathcal{X}$) is the only of such measures satisfying $SI^+$.

\begin{proposition}\label{SI+.convex}
Let $\rho$ be a coherent, law invariant, comonotonic additive risk measure. Then $\rho$ satisfies $SI^+$ if and only if $\rho(X) = ML(X)$.
\end{proposition}
\begin{proof}
We prove the ``only if" part only. The proof is based on the following representation of law inariant comonotonic additive coherent risk measures:
\begin{equation}
\rho(X)=\int_{[0,1]}\AVaR_\alpha(X)\mu(d \alpha)=:\rho_{\mu}(X),\quad \forall X \in \mathcal{X}.
\end{equation}
Assume that $\rho_{\mu}$ is not the $ML$. This amounts to saying that $\exists A \in \mathcal{B}([0,1])$ such that $\inf A >0$ and $\mu(A)>0$.

To see that $\rho_{\mu}$ cannot be $SI^+$, pick $X \in \mathcal{X}$ such that $\Prob(X \le 0)<\inf A$ and $\rho_{\mu}(X)\ge 0$. As in the previous proof, it suffices to show that $\rho_{\mu}(X^+)<0$. To that end, denote
\begin{equation}
B=\int_0^{\Prob(X^+ =0)} \AVaR_\alpha(X^+)\mu(d\alpha),
\end{equation}
and notice that $\rho_{\mu}(X^+) \le B$. To conclude the proof, it suffices to show that $B<0$.

Notice that, for any $\alpha > \Prob(X^+ =0)$ it holds $q_{X^+}(\alpha)>0$ and, therefore, $\AVaR_{\alpha}(X^+)<0$. Since $\Prob(X \le 0) <\inf A$, it holds that
\begin{equation}
A \subseteq (\Prob(X^+ =0),1] \text{ and, therefore, } \mu((\Prob(X^+ =0),1])>0,
\end{equation}
from which $B<0$ follows.
\end{proof}
The arguments we employed in the proofs of propositions \ref{SI+} and \ref{SI+.convex} are original. However, those results were already essentially know, as one can obtain a result very similar to \Cref{SI+} through the findings presented in \cite{he2018}.

\begin{proposition}\label{prop.excess.new}
Let $\rho\colon \mathcal{X}\rightarrow \mathbb{R}$ be a monetary, law invariant, and positive homogeneous risk measure. If $\rho$ is $SI^+$, then one the following holds:
\begin{enumerate}
\item $\rho(X) = \VaR_{p}(X)$ for all $X\in\mathcal{X}$ and some $p \in [0,1]$.
\item For all $X \in \mathcal{X}$, $q_{-X}^{+}(1-p)\le \rho(X)$ and $\rho(X)\le q_{-X}^{+}(1-p+\delta)$, for some $\delta \in (0,p)$.
\end{enumerate}
\end{proposition}
\begin{proof}
Since $\rho$ is law invariant, so is $\mathcal{A}_{\rho}$. Also, since $\rho$ is positive homogeneous and non-zero, the acceptance set $\mathcal{A}_{\rho}$ is a cone (see \Cref{prop_relations}). Now, if $\rho$ is $SI^+$, it follows by Proposition 4 in \cite{gao2020surplus} that $\mathcal{A}_{\rho}$ is surplus invariant. Theorem 1 in \cite{he2018} then implies one of the two cases:
\begin{enumerate}
\item $\mathcal{A}_{\rho}=\mathcal{A}_{\VaR_{p}}$ for some $p \in [0,1]$, in which case $\rho(X)=\VaR_{p}(X),\; \forall X \in \mathcal{X}$, or
\item $\{X \in \mathcal{X}:q_{-X}^{+}(1-p+\delta) \le 0,\text{ for some }\delta \in (0,p)\} \subseteq \mathcal{A}_{\rho}\\
 \quad \quad \quad  \subseteq \{X \in \mathcal{X}:q_{-X}^{+}(1-p) \le 0\}$.
\end{enumerate}
To conclude the proof, notice that the first case of the enumeration above corresponds to the first possibility enumerated in the statement of the theorem, and the second case above corresponds to the second possibility in the statement of the theorem.
\end{proof}
\section{Risky Eligible Assets}\label{eligible.1}

Most of the literature on risk measures assumes, if only for the sake of simplicity, the existence of a (unique) risk-free asset in which a financial institution bearing the risk $X\in \mathcal{X}$ can invest the regulatory capital $\rho(X)$ in order to meet the regulator's criteria of acceptability. This assumption is helpful for theoretical purposes but not realistic in some contexts, for instance during a period of financial crisis. In the absence of such an asset, one is led to work with random or ambiguous interest rates \citep{el2009cash}. Additionally, if there are more than one risk-free asset (e.g., for a financial institution with assets and liabilities denominated in different currencies) one is left with the problem of deciding which one to invest in \citep{artzner2009}.
Circumstances of this sort motivate the study of capital regulation in contexts where financial institutions are allowed to compose their regulatory capital with risky assets. In fact, this scenario is in accordance with the regulatory framework proposed by the Basel Committee on Banking Supervision \citep{bcbs2019feb}, which permits banks to compose their regulatory reserves with assets in different classes of risk.

An \defin{eligible asset} is a couple $S=(S_{0},\,S_{1})\in (0,\,\infty)\times L^{\infty}_{+}(\Omega,\mathcal{F},\Prob)$. The $S_{0}$ component is a constant representing the time $t=0$---today's---value of the asset, and the component $S_{1}$ represents its terminal payoff. If $S_{1}$ is non-constant, we say that $S$ is a risky eligible asset; otherwise $S$ is riskless, or risk-free. In what follows, we always assume that $S_{1}$ is bounded away from zero, i.e., $S_{1}\ge\epsilon$ for some $\epsilon>0$, and throughout this section we also make the assumption that acceptance sets are closed and monotone, following \cite{koch2018}. The combination of an acceptance set $\mathcal{A}$ and an eligible asset $S$ defines, for $X\in\mathcal{X}$, a risk measure $\rho_{\mathcal{A},S}$ through
\begin{equation}\label{rho_AS.1}
\rho_{\mathcal{A},S}(X)=\inf\left\{m \in \mathbb{R}\colon X+\frac{m}{S_{0}}S_{1} \in \mathcal{A}\right\}.
\end{equation}
\begin{remark}
In the zero interest rate framework (see \Cref{basic.1.appendix}), there is no difference between taking the random variables $X \in \mathcal{X}$ as terminal or discounted payoffs. For most of this paper, we work with discounted payoffs, however, when risky eligible assets are considered, it is convenient to let the random variables $X \in \mathcal{X}$ represent terminal payoffs. In view of this remark, notice that both $X$ and $(m/S_{0})S_{1}$ stand for financial positions expressed in the same (terminal) monetary unit. The intuition behind \cref{rho_AS.1} is that the financial institution will invest $m$ dollars to acquire $m/S_{0}$ units of the eligible asset $S$ which has terminal payoff $S_{1}$.
\end{remark}

\begin{remark}
As showed in Proposition 2.12 of \cite{farkas2014capital}, the properties we assume for $\mathcal{A}$ and $S$ implies that $\rho_{\mathcal{A},S}$ is Lipschitz continuous and finite. Under these hypothesis, the functional $\rho_{\mathcal{A},S}$ satisfies, for all $X\in\mathcal{X}$ and all $\lambda\in\mathbb{R}$,
\begin{itemize}
\item \label{s.add} \textit{(S-additivity)} $\rho_{\mathcal{A},S}(X+\lambda S_{1})=\rho_{\mathcal{A},S}(X)-\lambda S_{0}$,
\end{itemize}
meaning that the risk is equivariant with respect to the amount invested in the eligible asset. In the special case of $S=(1,1)$, we recover the traditional cash additivity property.
\end{remark}

\subsection{Comonotonicity and Risky Eligible Assets}\label{general.koch}

\cite{koch2018} provide necessary and sufficient conditions for $\rho_{\mathcal{A},S}$ to be comonotonic additive. As the following proposition shows, these conditions are quite restrictive as they require the regulator to deem acceptable even highly leveraged financial positions. Notice that, when specialized to a slightly less general case, their result shows that the usage of comonotonic additive risk measures is incompatible with the usage of risky eligible assets.

\begin{proposition}\label{koch.1}
(\cite{koch2018} - Theorem 2.18, Corollary 2.20) Assume that $\rho_{\mathcal{A}}$ is comonotonic additive. Then, the following statements are equivalent:
\begin{enumerate}[noitemsep,nosep]
\item $\rho_{\mathcal{A},S}$ is comonotonic additive.
\item\label{item2-koch.1} $\mathcal{A}\pm\left(1+({\rho_{\mathcal{A},S}(1)}S_{1}/S_0)\right)\subset \mathcal{A}$.
\end{enumerate}
Moreover, if $\mathcal{A}\cap(-\mathcal{A})=\{0\}$, and if $S$ is a risky eligible asset, then $\rho_{\mathcal{A},S}$ is not comonotonic additive.
\end{proposition}
\begin{remark}
As showed in \cite{koch2018} the condition $\mathcal{A}\cap (-\mathcal{A})=\{0\}$ holds for the VaR, AVaR, and all spectral risk measures. However, the authors showed that, as a consequence of this condition, comonotonic additivity of those representative risk measures is lost once risky eligible assets are considered (see propositions \ref{koch.SVaR} and \ref{koch.S-DR} below).
\end{remark}
As suggested in \cite{koch2018}, taking $\rho_{\mathcal{A},S}(1)=-1$ allows us to better understand why \Cref{item2-koch.1} in \Cref{koch.1} is restrictive.
In fact, this normalization---assuming $\rho_{\mathcal{A},S}$ is comonotonic additive and $0\in\mathcal{A}$---yields
\begin{math}
\pm\left(1-({S_{1}}/{S_{0}})\right)\subset \mathcal{A},
\end{math}
with the random variable $1-(S_{1}/S_{0})$ representing the position of a bank that financed one unit of the risk-free asset by short-selling $1/S_{0}$ units of $S$.\footnote{Ignoring transaction costs.} This position realizes losses exactly when $S_{1}>S_{0}$. Therefore, if the eligible asset pays positive interests with probability one, which amounts to $\Prob(S_{1}>S_{0})=1,$ then the position $1-(S_{1}/S_{0})$ realizes losses with probability one and, nonetheless, is acceptable. As the following corollary shows, this reasoning holds for any multiple of $1-(S_1/S_0)$.

\begin{corollary}\label{prop.excess.2}
Let $\mathcal{A}$ be a monetary acceptance set such that $0 \in \mathcal{A}$. If $\rho_{\mathcal{A},S}$ is comonotonic additive with $\rho_{\mathcal{A},S}(1)=-1$, then
\begin{math}
\operatorname{span}\left(1-({S_{1}}/{S_{0}})\right) \subset \mathcal{A}.
\end{math}
\end{corollary}
\begin{proof}
In view of the relation $\pm\left(1-({S_{1}}/{S_{0}})\right)\subset \mathcal{A}$, it suffices to show that $\mathcal{A}$ is a cone. To see that this is the case, notice that the comonotonicity of $\rho_{\mathcal{A},S}$ implies that of $\rho_{\mathcal{A}}$ (see Proposition 2.15 of \cite{koch2018}). The monetarity of $\mathcal{A}$ implies that of $\rho_{\mathcal{A}}$ (see \cref{item2.prop_relations} of \Cref{prop_relations}). Also, the cash additivity of $\rho_{\mathcal{A}}$ implies that it is non-zero. Therefore we can apply Proposition 2.5 of \cite{koch2018} to conclude that $\rho_{\mathcal{A}}$ is positive homogeneous. Then $\mathcal{A}_{\rho_{\mathcal{A}}}$ is a cone (see \cref{item5.prop_relations} of \Cref{prop_relations}) and, therefore, $\mathcal{A}$ is a cone (see \cref{item3.prop_relations} of \Cref{prop_relations}).
\end{proof}
\begin{remark}\label{extends}
The above corollary summarizes the discussion \cite{koch2018} presented after their Proposition 2.18. Essentially, it says that, if an acceptance set $\mathcal{A}$ and a risky eligible asset $S$ together induce a comonotonic additive risk measure, then $\mathcal{A}$ must contain arbitrarily large fully leveraged positions. In the following subsection, we illustrate a consequence of this result by comparing the risk of $1-(S_{1}/S_{0})$ as measured by the traditional value-at-risk (which is comonotonic additive), vis-à-vis the value-at-risk based on a risky eligible asset (which, it turns out, is not comonotonic additive).
\end{remark}

\subsection{Examples}\label{specific.koch}
In this section, we review the findings of \cite{koch2018} regarding the lack of comonotonic additivity of particular risk measures based on risky eligible assets. \cite{koch2018} constructed ``risky eligible" counterparts of VaR, AVaR (examples \ref{example.var} and \ref{example.avar}), and of the class of distortion risk measures. All these risk measures are, originally, comonotonic additive. However, their counterparts inherit the property of comonotonic additivity if and only if the eligible asset being used is risk-free.

\cite{koch2018} defined the counterpart of the value-at-risk with respect to an eligible asset $S=(S_{0},S_{1})$ through $\SVaR_{p} \defcol \rho_{\mathcal{A},S}$ with $\mathcal{A} \defcol \mathcal{A}_{\VaR_{p}}$---see Equation \eqref{rho_AS.1}.
\begin{proposition}\label{koch.SVaR}
(\cite{koch2018} - Proposition 3.4) The risk measure $\SVaR_{p}$ is comonotonic if and only if $S$ is risk-free.
\end{proposition}
It follows from this proposition that, if the regulatory authority insists on a comonotonic additive risk measure, then it cannot determine the banks' regulatory capital through a rule of the type ``banks should invest monetary amounts (m) in the asset $S$ until $\Prob(X+(mS_{1}/S_{0})<0)\le p$", where $p$ is usually taken as $0.01$.
Results of the same nature as \Cref{koch.SVaR} were also obtained for the AVaR and the class of distortion risk measures, the latter being of the form $\operatorname{DR}_{\mu}(X)\coloneqq\int_{0}^{1}\AVaR_{p}(X)\mu(\mathrm{d}p)$, where $p\in(0,1]$ and $\mu$ is a probability measure on the Borel sets of $[0,1]$\footnote{We present these risk measures in item 2 of \Cref{kusuoka_com.1}.}. 

In addition, \cite{koch2018} introduced the risk measures $\SAVaR_{p}\coloneqq\rho_{\mathcal{A}_{\AVaR_{p}},S}$ and $\operatorname{S-DR}_{\mu}\coloneqq\rho_{\mathcal{A}_{\operatorname{DR}_{\mu}},S}$, for which the following holds.
\begin{proposition}\label{koch.S-DR}
(\cite{koch2018} - Propositions 3.7 and 3.10) The risk measure $\SAVaR_{p}$ is comonotonic additive if and only if $S$ is risk-free. Also, the risk measure $\operatorname{S-DR}_{\mu}$ is comonotonic additive if and only if one of the following conditions holds:
\begin{enumerate}[noitemsep,nosep]
\item $\mu(\{1\})=1$ (so that $\operatorname{DR}_{\mu}(X)=-\E[X]$ for all $X \in \mathcal{X}$).
\item $S$ is risk-free.
\end{enumerate}
\end{proposition}

We conclude this section by showing that, if the eligible asset is risky, then the risk assessments obtained through a comonotonic additive risk measure of the form $\rho_{\mathcal{A},S}$ can be drastically different from those obtained through the more traditional (non-comonotonic additive) $\SVaR$ and $\SAVaR$.
\begin{corollary}
Let $\mathcal{A}$ be a monetary acceptance set and $S$ a risky eligible asset such that $\rho_{\mathcal{A},S}(1)=-1$.
\begin{enumerate}[noitemsep,nosep]
\item If $\rho_{\mathcal{A},S}$ is comonotonic additive and $0 \in \mathcal{A}$, then
\begin{math}
\rho_{\mathcal{A},S}\left(\lambda \left(1-{S_{1}}/{S_{0}}\right)\right)\le 0
\end{math}
holds for all $\lambda\ge0$.
\item If $\Prob(S_{0}<S_{1})>p$, for $p \in (0,1)$, then
\begin{equation}
\lim_{\lambda \rightarrow \infty}\SVaR_{p}\left(\lambda \left(1-\frac{S_{1}}{S_{0}}\right)\right)=\lim_{\lambda \rightarrow \infty}\SAVaR_{p}\left(\lambda \left(1-\frac{S_{1}}{S_{0}}\right)\right)=\infty.
\end{equation}
\end{enumerate}
\end{corollary}
\begin{proof}
The first item is a direct consequence of the relation $\pm\left(1-{S_{1}}/{S_{0}}\right)\subset \mathcal{A}$. The second item follows for $\Prob(S_{0}<S_{1})>p$ implies $\Prob(1-(S_{1}/S_{0})<0)>p$ and, therefore, $\SVaR_{p}(1-(S_{1}/S_{0}))>0$. The conclusion follows by the positive homogeneity of $\SVaR$ and the fact that $\SVaR(X)\le\SAVaR(X)$ for all $X \in \mathcal{X}$.
\end{proof}

\section{Elicitability}\label{elicit}

In the last decade, elicitability of risk measures has become a prominent research agenda in the literature \citep{kou2016,bellini2015_OnElicitableRiskMeasures,ziegel2016,acerbi2017,fissler2021elicitability,nolde2017elicitability,delbaen2016risk,he2022risk,embrechts2021bayes}. Arguably, the reason for such interest is that elicitability is a property that allows for meaningful comparison of the predictive performance of competing forecasting procedures, and such comparisons are especially important for risk management because the tails of the distributions are difficult to estimate \citep{kou2016}. Additionally, elicitability permits well grounded inference procedures \citep{gneiting2011}. There is compelling evidence for the importance of elicitability, as shown in \cite{gneiting2011}, \cite{patton2011}, in the supplementary material of \cite{nolde2017elicitability}, and in \cite{fissler2021elicitability}, to name a few. In this section, we review some results of \cite{bellini2015_OnElicitableRiskMeasures}, \cite{kou2016}, and \cite{ziegel2016} unveiling the scarcity of comonotonic additive, elicitable risk measures.

Determining the regulatory capital $\rho(X)$ for a position $X \in \mathcal{X}$ in practice requires estimation of the distribution $F_X$.\footnote{Alternatively, in Bayesian/subjectivist approaches one consider \textit{personal credences} to quantify risk.} We assume law invariance throughout this section as it is a fundamental property in an applied setting. Every law invariant risk measure on $\mathcal{X} = L^\infty$ induces a functional mapping the set of probability distributions with bounded support---denote it by $\mathcal{P}\coloneqq\{F_{X}\colon X \in \mathcal{X}\}$---into real valued risk assessments.\footnote{The results of \cite{kou2016} were obtained for more general domains. We keep with $\mathcal{X}=L^{\infty}$ for the sake of unity with the rest of the paper.}
It is convenient not to change the symbol used to denote the induced ``statistical" risk measures and, therefore, these are defined, for $F\in\mathcal{P}$, as
\begin{math}
\rho(F)\coloneqq\rho(X),
\end{math}
for $X \in \mathcal{X}$ such that $F_X=F$.

The criterion to rank two alternative forecasting procedures, say $A$ and $B$, that produce theoretical forecasts $x^{A}$ and $x^{B}$ for the true value $\rho(F_{X})$, follows rules of the type:
\begin{equation}\label{crit.elicit}
``A\text{ is better than }B\text{ if and only if }\E[S(x^{A},X)]\le\E[S(x^{B},X)]",
\end{equation}
where $S\colon \mathbb{R}^{2}\rightarrow \mathbb{R}_{+}$ is a non-negative function such that $S(x,y)$ is (usually) increasing in the difference $|x-y|$. There are two main assumptions behind such rules:
\begin{enumerate}[noitemsep,nosep]
\item For all $X \in \mathcal{X}$, $\rho(F_{X})$ minimizes $\E[S(x,X)]$ w.r.t.\ $x \in \mathbb{R}$, and;
\item If $x^{A},x^{B}\in \mathbb{R}$ are such that $\rho(F_{X})<x^{A}<x^{B} \text{ or } x^{B}<x^{A}<\rho(F_{X}), \text{ then } \E[S(x^{A},X)]\le\E[S(x^{B},X)]$.
\end{enumerate}
The first condition guarantees that the criterion in \cref{crit.elicit} ranks an estimation procedure producing the true value above any other procedure. The second condition guarantees that the ranking produced by the criterion in \cref{crit.elicit} is consistent with the results produced by the competing estimation procedures. This second condition received special attention in \cite{bellini2015_OnElicitableRiskMeasures}.

There are slight differences in the literature regarding the formal definition of elicitability. The following definition, for instance, does not require the second condition.
\begin{definition}\label{elic.def}
(\cite{kou2016}) A single-valued statistical functional $\rho\colon \mathcal{P}\rightarrow \mathbb{R}$ is \defin{general elicitable} with respect to a class of distributions $\mathcal{M}\subseteq \mathcal{P}$ if there exists a scoring function $S\colon \mathbb{R}^{2} \rightarrow \mathbb{R}$ such that, for all $X\in\mathcal{X}$ for which $F_X\in\mathcal{M}$, it holds that
\begin{math}
\rho(F_X)=-\min\big\{r\in\mathbb{R}\,\big\vert\,r \in \argmin_{x\in\mathbb{R}}\E[S(x,X)]\big\}.
\end{math}
In this case, we say that $S$ is \defin{consistent} for $\rho$ with respect to the class $\mathcal{M}$.
\end{definition}
\Cref{elic.def} draws from the intuition that, when it comes to the estimation of a given (elicitable) functional and, more specifically, to the evaluation of estimation procedures, there should be a match between the functional being estimated, on the one hand, and the corresponding score function being used, on the other. In this regard, \cite{gneiting2011} presented a compelling argument showing that a mismatch between $S$ and $\rho$ can lead to twisted decisions regarding the relative performance of alternative estimation procedures.
\begin{example}
The squared deviation score $S(x,y)=(x-y)^{2}$ is consistent for $\rho(\cdot)=\mathbb{-E}[\cdot]$ with respect to the class of distributions with finite first moment. This happens because, for $X$ with finite first moment, it follows that $-\E[X]=-\argmin_{x \in \mathbb{R}}\E[(x-X)^{2}]$. As documented in \cite{gneiting2011}, the squared deviation function is, by far, the most used in academia and industry.
\end{example}
\begin{example}
The function
\begin{math}
S(x,y)=(1(x\ge y)-p)(g(x)-g(y)),
\end{math}
where $1(\cdot)$ denotes the indicator functions and $g\colon \mathbb{R}\rightarrow \mathbb{R}$ is increasing, is consistent for the value-at-risk with respect to the class of distributions with finite first moment.
\end{example}
\begin{remark}
A necessary condition for a statistical functional $\rho$ to be elicitable with respect to a given class $\mathcal{P}$ is that
the level sets $\{F\in\mathcal{P}\colon \rho(F) = r\}$ are convex, for all $r\in\mathbb{R}$ \citep{weber2006}.
For a detailed study of such risk measures see \cite{delbaen2016risk}. Also, it is valid to observe that not being elicitable with respect to a class $\mathcal{P}_{0}\subseteq \mathcal{P}$ implies not being elicitable with respect to $\mathcal{P}$. This fact follows for, if the level sets of $\rho$ are not convex within $\mathcal{P}_{0}$, then they are not convex when the larger class $\mathcal{P}$ is considered. In some cases, this observation allows one to restrict attention to elicitability with respect to probability distributions with finite support \citep{kou2016}.
\end{remark}
\cite{weber2006}, \cite{gneiting2011}, \cite{bellini2015_OnElicitableRiskMeasures}, \cite{ziegel2016}, and \cite{kou2016} are main references characterizing classes of elicitable risk measures. As it turns out, elicitability for risk measures is the exception rather than the rule. For instance, AVaR is not elicitable \citep{weber2006,gneiting2011}. Keeping the focus on coherent risk measures, \cite{bellini2015_OnElicitableRiskMeasures} and \cite{ziegel2016} showed (independently) that the class of coherent, elicitable risk measures consists of expectiles only.
\begin{definition}\label{def.expectiles}
(\cite{newey1987}, \cite{ziegel2016}) For $\tau \in (0,1)$ and $X \in \mathcal{X}$, the $\tau$-expectile of $X$, denoted $\mu_{\tau}(X)$, is the unique solution to the following equation:
\begin{equation}
\tau \int_{x}^{\infty}(y-x)F_{X}(\mathrm{d}y)=(1-\tau)\int_{-\infty}^{x}(x-y)\,F_{X}(\mathrm{d}y).
\end{equation}
\end{definition}

\begin{theorem}\label{teo.ziegel}
(\cite{ziegel2016} - Corollaries 4.3 and 4.6) Let $\rho\colon \mathcal{X}\rightarrow \mathbb{R}$ be a monetary, law invariant risk measure whose statistical counterpart is elicitable with respect to any class of probability distributions that contains the two-point distributions. Then
\begin{enumerate}[noitemsep,nosep]
\item $\rho$ is coherent if and only if $\rho(X)=-\mu_{\tau}(X)$, $\forall X \in \mathcal{X}$, for some $\tau \in (0,1/2]$.
\item \label{item.b.ziegel}$\rho$ is coherent and comonotonic additive if and only if $\rho(X)=-\E[X],\; \forall X \in \mathcal{X}$.
\end{enumerate}
\end{theorem}
\begin{remark}
\Cref{teo.ziegel} tells us that, for applications in which elicitability is essential and coherence is desirable, one is bound to adopt an expectile as the risk measure. If, in addition, one demands comonotonic additivity, then the expected loss is the only option.
\end{remark}
\begin{theorem}\label{teo.kou.elic}
(\cite{kou2016} - Theorem 1) Let $\rho\colon \mathcal{X}\rightarrow \mathbb{R}$ be a monetary, law invariant, comonotonic additive risk measure. Then the statistical counterpart of $\rho$ is elicitable with respect to the class of discrete distributions if and only if, for all $X \in \mathcal{X}$, either $\rho(X)=-\E[X]$ or $\rho(X)=\VaR_{p}(X)$ for some $p\in[0,1]$.
\end{theorem}
\begin{remark}
\Cref{teo.kou.elic} complements \Cref{teo.ziegel} by showing that, even outside the coherent framework, the axiom of comonotonic additivity considerably narrows the class of elicitable risk measures (this result was corroborated in \cite{wang2015}).
\end{remark}

Before concluding, we must mention the work of \cite{fissler2016higher}. They generalized the concept of elicitability, extending it to vector-valued functionals. In this case, one says that the components of the vector-valued functional are jointly-elicitable. As much as for real-valued functionals, jointly elicitability gives us a method to compare the performance of alternative forecast procedures. Remarkably, this can be done even if the components of the vector-valued functional are not individually elicitable. A prominent example of this is the functional $T(X)\defcol (\AVaR_{p}(X),\VaR_{p}(X))$, which is jointly-elicitable, even if $\AVaR_{p}$ is not elicitable as a standalone risk measure. Also, \cite{fissler2016higher} showed that any (finite) convex combination of $\AVaR_{p}$ (for significance levels $0< p_{0}<p_{1}<\cdots<p_{n}\le1$) is jointly-elicitable with the quantiles $p_{0}<p_{1}<\cdots<p_{n}$. These convex combinations are coherent risk measures, which form a (narrow) subclass of spectral risk measures (see \Cref{kusuoka_com.1}).

\section{Time-consistency}\label{section_dynamics}

The most prominent benefit of generalizing risk measures to the dynamic context is to allow the risk to depend on the flow of arriving information. For a given probability space $(\Omega,\mathcal{F},\Prob)$, the information flow is modeled through a filtration $(\mathcal{F}_{t})_{t \in \mathcal{T}}$. The $\sigma$-algebra $\mathcal{F}_{t}$ represents the information available at time $t \in \mathcal{T}$, and the time horizon $\mathcal{T}$ may be discrete ($\mathcal{T}\coloneqq\{0,1,2,\dots,T\}$) or continuous ($\mathcal{T}=[0,T]$), and might be finite ($T \in \mathbb{R}$) or infinite ($T=\infty$). To simplify the exposition we restrict our attention to the discrete, finite case. We denote $L^{\infty}_{t}\coloneqq L^{\infty}(\Omega,\mathcal{F}_{t},\Prob)$, and assume $\mathcal{F}_0 = \{\varnothing,\Omega\}$ and $\mathcal{F}_{T}=\mathcal{F}$. Therefore we have $L^{\infty}_{T}= L^{\infty}\coloneqq L^{\infty}(\Omega,\mathcal{F},\Prob)$. To measure the risk of a financial position $X\in L^{\infty}$ conditional on the information available at $t \in \mathcal{T}$ one usually relies on the following tools:
\begin{definition}\label{def.conditional}
(\cite{follmer2016stochastic}, \cite{delbaen2021commonotonicity}) For $t \in \mathcal{T}$, we call a map $\rho_{t}\colon L^{\infty}\rightarrow L^{\infty}_{t}$ a \defin{conditional risk measure}. Also, we call $\rho_{t}$ a \defin{monetary conditional risk measure} if it satisfies the following properties:
\begin{enumerate}[noitemsep,nosep,series=cond.prop]
\item (Conditional Cash Additivity) $\rho_{t}$ is \defin{conditionally cash additive} if $\rho_{t}(X+Z)=\rho_{t}(X)-Z$ for any $X \in L^{\infty}$ and $Z\in L_{t}^{\infty}$.
\item (Monotonicity) $\rho_{t}$ is \defin{monotone} if $X\le Y$ implies $\rho_{t}(Y)\le \rho_{t}(X)$ for all $X,Y \in L^{\infty}$.
\item (Normalization) $\rho_{t}$ is \defin{normalized} if $\rho_{t}(0)=0$.
\end{enumerate}
In addition, a conditional risk measure might satisfy
\begin{enumerate}[noitemsep,nosep,resume = cond.prop]
\item (Conditional Comonotonicity) $\rho_{t}$ is \defin{conditionally comonotonic} if $\rho_{t}(X+Y)=\rho_{t}(X)+\rho_{t}(Y)$ for all comonotonic $(X,Y) \in L^\infty\times L^\infty$.
\end{enumerate}
\end{definition}
Notice that the static framework is recovered by the risk measure $\rho_{0}\colon L^{\infty}\rightarrow L^{\infty}_{0}=\mathbb{R}$. Conditional risk measures generalize this static perspective on risk, allowing us to measure, at $t=0$, the abstract notion of the ``risk of $X \in L^\infty$ at $t>0$". In the same vein as conditional expectations, these conditional risk measurements are random variables whose distribution depends on the filtration. To illustrate this analogy, notice that conditional cash additivity and normalization implies that $\rho_{T}(X)=-X$ for all $X \in L^{\infty}$, which is (up to the sign) the result of taking expectation w.r.t. $\mathcal{F}$.

The traditional properties of convexity, positive homogeneity, and subadditivity also have counterparts for conditional risk measures, and most of the basic theory presented in the Appendix's section \ref{basic.1.appendix} can be adapted to conditional risk measures (see \cite{acciaio2011} and \cite{follmer2016stochastic} for details).
\begin{definition}
A collection $(\rho_{t})_{t\in \mathcal{T}}$ is called a \defin{dynamic risk measure} if $\rho_{t}$ is a conditional risk measure for each $t \in \mathcal{T}$. The following are properties that a dynamic risk measure might satisfy:
\begin{enumerate}[noitemsep,nosep]
\item (Time-consistency) $(\rho_{t})_{t\in \mathcal{T}}$ is \defin{time-consistent} if, for all $t\in \{0,1,\dots,T-1\}$ and all $X,Y\in L^\infty$ satisfying $\rho_{t+1}(X)\ge \rho_{t+1}(Y)$, it holds that 
\begin{math}
\rho_{t}(X)\ge \rho_{t}(Y).
\end{math}
\item (Relevance) $(\rho_{t})_{t\in \mathcal{T}}$ is \defin{relevant} if $\rho_{0}(-\epsilon1_{A})>0$ for all $\epsilon>0$ and all $A \in \mathcal{F}$ such that $\Prob(A)>0$.
\end{enumerate}
\end{definition}
\begin{remark}
We say that a dynamic risk measure $(\rho_{t})_{t \in \mathcal{T}}$ satisfies a property presented in \Cref{def.conditional} if the respective property holds for $\rho_{t}$ for all $t \in \mathcal{T}$. In particular, $(\rho_{t})_{t \in \mathcal{T}}$ is monetary if each $\rho_{t}$ is monetary.
\end{remark}
Arguably, the main concern about dynamic risk measures is to define how the risks in different periods should relate. For instance, consider two financial positions $X,Y$ in $L^{\infty}$ such that $X\le Y$. In this case, an investor in $t=0$ knows with certainty that at $t=T$ the result of $X$ will be worse than that of $Y$. With this in mind, the investor would know, at $t=0$, that, irrespectively of what might happens between $t=0$ and $t=T$, the risk of $X$ will be greater than that of $Y$ at $t=T$, i.e., $\rho_{T}(X)\ge \rho_{T}(Y)$. This follows by assuming that $\rho_{T}$ is monotone, which is a minimal assumption for risk measures. Now, if the investor knows that, at the end of the game, the risk of $X$ is greater than that of $Y$, then it would be ``reasonable" to use this information when comparing the risk of the positions at $t=T-1$. An iteration of this argument leads to time-consistency.

\begin{remark}
The time-consistency property can be equivalently defined in a manner similar to the “tower property” of conditional expectation: $\rho_{t}(X)=\rho_{t}(-\rho_{t+1}(X))$ for all $t \in \{0,1,\dots , T-1\}$. This condition illustrates that, in the time-consistent framework, the time $t$ risk of an $\mathcal{F}_{T}$ measurable random variable $X$ is fully determined by the random variable $\rho_{t+1}(X)$, which is measurable with respect to $\mathcal{F}_{t+1}$ (see, for instance, \cite{acciaio2011} and \cite{follmer2016stochastic} for details).

\end{remark}
\begin{theorem}\label{teo.kupper}
(\cite{kupper2009} - Theorem 1.10) Let $(\Omega, \mathcal{F},(\mathcal{F}_{t})_{t \in \mathcal{T}},\Prob)$ be a standard filtered probability space. A monetary dynamic risk measure $(\rho_{t})_{t \in \mathcal{T}}$ is time-consistent, relevant, and law invariant if and only if there is a $\beta \in (-\infty,\infty]$ such that the representation
\begin{equation}\label{eq:entropic}
\rho_{t}(X)=\frac{1}{\beta}\ln \E[e^{-\beta X}|\mathcal{F}_{t}]
\end{equation}
holds for all $t\in\mathcal{T}$ and all $X\in L^\infty$.
\end{theorem}
\Cref{teo.kupper} illustrates the scarcity of time-consistent dynamic risk measures. The mappings appearing in the representation in Equation~\eqref{eq:entropic} are called \defin{entropic} conditional risk measures with risk aversion parameter $\beta$. Since the static counterparts of entropic risk measures (namely, $\rho_0$) are not comonotonic additive, it follows that one cannot obtain a time-consistent dynamic risk measure whose conditional components are comonotonic additive.

The conflict between comonotonicity and time-consistency also appears in \cite{delbaen2021commonotonicity}. In his three period framework $\mathcal{T}=\{0,1,2\}$, a dynamic risk measure is a pair $(\rho_{0},\rho_{1})$, with $\rho_{2}$ existing only implicitly since its conditional cash additivity would imply $\rho_{2}(X)=-X$ for all $X \in L^{\infty}_{2}$. In this setting, the tower property requirement for time-consistency boils down to $\rho_{0}(X)=\rho_{0}(-\rho_{1}(X))$ for all $X \in L^{\infty}_{2}$.

\begin{definition}
Consider the following definitions:
\begin{enumerate}[noitemsep,nosep]
\item Let $\rho_{t}$ be a conditional risk measure for some $t \in \mathcal{T}$. We say $\rho_{t}$ is \defin{Lebesgue continuous} if, whenever $(X_{n}) \subseteq L^{\infty}$ is uniformly bounded and $X_{n}\rightarrow X$ in probability, we have $\rho_{t}(X_{n})\rightarrow \rho_{t}(X)$ in probability.
\item We say that $\mathcal{F}_{2}$ is \defin{atomless conditionally to $\mathcal{F}_{1}$} if, for every $A \in \mathcal{F}_{2}$, there exists a set $B \subseteq A$, $B \in \mathcal{F}_{2}$, such that $0<\E[1_{B}|\mathcal{F}_{1}]<\E[1_{A}|\mathcal{F}_{1}]$ on the set $\{\E[1_{A}|\mathcal{F}_{1}]>0\}$.
\end{enumerate}
\end{definition}
\begin{theorem}\label{teo.delbaen}
(\cite{delbaen2021commonotonicity} - Theorem 6.1) Assume that $\mathcal{F}_{2}$ is atomless conditionally to $\mathcal{F}_{1}$ and let $(\rho_{t})_{t \in \mathcal{T}}$ be a time-consistent dynamic risk measure. Also, assume that $\rho_{0}$ is coherent, relevant, comonotonic additive, and Lebesgue continuous. Then there is a probability measure $\ProbQ$, equivalent to $\Prob$, such that
\begin{math}
\rho_{0}(X)=\E_{\ProbQ}[-X] \text{ for all } X \in L^{\infty}(\mathcal{F}_{1}).
\end{math}
\end{theorem}
\begin{remark}
\Cref{teo.delbaen} shows that, by insisting on both comonotonic additivity and time-consistency, the set of coherent risk measures (satisfying the additional hypothesis of the theorem) collapses to an expected value. In comparison to \Cref{teo.kupper}, the conflict between comonotonic additivity and time-consistency presented in \Cref{teo.delbaen} is more direct. Also, \Cref{teo.delbaen} does not rely on law invariance, which called for different proof techniques and juxtapose \cite{delbaen2021commonotonicity} with the recent research on law invariant risk measures that collapses to the mean \citep{bellini2021law,liebrich2022law}.
\end{remark}

\section{Application: Portfolio Analysis}\label{port}

In Finance, risk measures are studied primarily as a tool for regulatory capital determination.
The notion of risk, however, is pervasive not only in Finance, but also in the Actuarial Sciences, in Decision Theory, and in many other fields. Thus, risk measures may be employed in a variety of theoretical and applied frameworks, even beyond the scope of determining regulatory capital.

In this section, we depart from the study of risk measures on their own right, and summarize some findings of \cite{brandtner2013conditional} regarding the usage of spectral risk measures in portfolio selection problems. This leads to two quirks, which are not present in the mean-variance framework, and that may impose extra difficulties to portfolio optimization: first, \cite{brandtner2013conditional} showed that the traditionally equivalent problems of, on the one hand, minimizing risk subject to a prespecified level of expected return and, on the other hand, maximizing a utility function that balances the trade-off between risk and return, are no longer equivalent.
A second peculiarity that accompanies the usage of spectral risk measures in portfolio selection is that the solutions to the associated optimization problems turn out to be corner solutions. In particular, when the risk-free asset is included in the analysis, these corner solutions correspond to invest all or nothing in the risk-free asset or in the tangency portfolio. Therefore, if short-sales are allowed, using spectral risk measures for portfolio selection may involve assuming extremely leveraged positions.

We consider two risky assets, $X_{1},X_{2} \in \mathcal{X}$, and a risk-free asset, $X_{0} \in \mathbb{R}$. We refer to \cite{brandtner2013conditional} for the extension to the case of a finite general number of risky assets. The set of possible portfolios is $\mathcal{X}^{\star}=\{\beta(\gamma X_{1}+(1-\gamma)X_{2})+(1-\beta)X_{0}\colon \beta \ge 0,\,\gamma \in \mathbb{R}\}$. A typical element of $\mathcal{X}^{\star}$ is denoted as $X_{\beta,\gamma}$.

\subsection{Mean-variance portfolio analysis}\label{nonequiv}
Let $\V$ be the variance operator and consider the two following problems:
\begin{enumerate}[noitemsep,nosep]
\item Limited analysis:
\begin{align}
\begin{split}\label{limit.var1-2}
&\min_{\beta \ge 0, \gamma\in \mathbb{R}}\V(X_{\beta,\gamma})\\
&s.t.: \E[X_{\beta,\gamma}]=\mu
\end{split}
\end{align}
\item Trade-off analysis:
\begin{equation}\label{mvutility}
\max_{\beta \ge 0,\gamma \in \mathbb{R}}\E[X_{\beta,\gamma}]-\frac{\lambda}{2}\V(X_{\beta,\gamma}).
\end{equation}
\end{enumerate}
\begin{remark}
The level $\mu \in \mathbb{R}$ in \cref{limit.var1-2} is usually required to be greater than the expected return of the minimum variance portfolio. In the absence of this restriction, the minimum variance portfolio is the obvious solution.
The trade-off analysis given in \cref{mvutility} has a strong theoretical basis for the case where the investor's absolute risk aversion is constant (and equal to $\lambda$) and the return of the risky assets is normally distributed \citep{bamberg1986hybrid}. The limited analysis, on the other hand, might be more adequate for applications where the return level, $\mu$, is determined at a higher hierarchical level of the financial analysis, so that the portfolio manager is restricted to portfolios with a mean return equal to $\mu$.
\end{remark}
For the next proposition, $\gamma_{MVP}$ denotes the weight for the minimum variance portfolio, and $X_{T,\sigma^{2}}$ denotes the tangency portfolio of the $(\mu,\sigma^{2})$-analysis. To focus on the main message, we refer the reader interested in the exact expressions for $\gamma_{MVP}$ and $X_{T,\sigma^{2}}$ to the original article.
\begin{proposition}\label{prop.corner.var}
(\cite{brandtner2013conditional}-Proposition 4.2)
The solution to the $(\mu,\sigma^{2})$ trade-off analysis, eq. \eqref{mvutility}, is given by
\begin{equation}\label{brand.var}
\gamma^{\star}=\gamma_{MVP}-\frac{\E[X_{2}-X_{1}]}{\lambda (\V(X_{1})+\V(X_{2})-2\Cov(X_{1},X_{2}))}
\end{equation}
\begin{equation}\label{brand.2}
\beta^{\star}=\frac{\E[X_{T,\sigma^{2}}-X_{0}]}{\lambda \V(X_{T,\sigma^{2}})}.
\end{equation}
\end{proposition}
The limited and the trade-off analysis approaches are equivalent in the mean-variance framework: there exists a one-to-one correspondence between the parameters $\mu$ and $\lambda$ such that the problems \eqref{limit.var1-2} and \eqref{mvutility} generate the same solution whenever $\mu$ is chosen as $\mu(\lambda)$ or, equivalently, $\lambda=\lambda(\mu)$.\footnote{The specific form of the correspondence $\mu \leftrightarrow \lambda$ can be found in \cite{brandtner2013conditional} and \cite{steinbach2001markowitz}.} This equivalence, however, does not hold if the variance is replaced by a spectral risk measure.

Also, notice that item 1 above shows that the solution to the mean-variance problem does not lie in the corner, i.e., $\gamma^{\star}$ and $\beta^{\star}$ are finite and are different from $0$, except for very specific cases. Also, notice that the optimal proportion invested in risky assets, the parameter $\beta^{\star}$ in \cref{brand.2}, is proportional to the risk-adjusted return of the tangency portfolio.

\subsection{Portfolio analysis with spectral utilities}
As shown in \Cref{spectral.brend}, a risk measure $\rho\colon\mathcal{X}\rightarrow \mathbb{R}$ is coherent, comonotonic additive, law invariant, and continuous if and only if
\begin{math}
\rho(X)=-\int_{0}^{1}\phi(t)\cdot\inf\{x \in \mathbb{R}\colon F_{X}(x)\ge t\}\,\mathrm{d}t,
\end{math}
$X\in\mathcal{X}$, for some non-negative, decreasing probability density $\phi\colon[0,1]\to\mathbb{R}_{+}.$ Write $\rho = \rho_\phi$ for such a risk measure, and consider the following counterparts of the mean-variance problems with the variance replaced by $\rho_{\phi}$:
\begin{enumerate}[noitemsep,nosep]
\item Limited analysis
\begin{align}
\begin{split}\label{spec.lim1-2}
&\min_{\beta \ge 0, \gamma \in \mathbb{R}}\rho_{\phi}(X_{\beta,\gamma})\\
&s.t.:\E[X_{\beta,\gamma}]=\mu
\end{split}
\end{align}
\item Trade-off analysis
\begin{equation}\label{spec.tradeoff}
\max_{\beta \ge 0, \gamma\in \mathbb{R}}(1-\lambda)\E[X_{\beta,\gamma}]-\lambda\rho_{\phi}(X_{\beta,\gamma}),\quad \lambda \in [0,1]
\end{equation}
\end{enumerate}

\begin{proposition}\label{prop.corner}
(\cite{brandtner2013conditional}-Proposition 4.3) The following items give the solutions for the problem in \cref{spec.tradeoff} when short-sales are allowed and restricted, respectively.
\begin{enumerate}[noitemsep,nosep]
\item The solution to the $(\mu,\rho_{\phi})$ trade-off analysis, \cref{spec.tradeoff} is given by
\begin{equation}\label{brand.3}
\beta^{\star}=\begin{cases}
0,&\text{ if } \frac{\E[X_{T,\rho_{\phi}}-X_{0}]}{\rho_{\phi}(X_{T,\rho_{\phi}})-\rho_{\phi}(X_{0})}\le \frac{\lambda}{1-\lambda}\\
+\infty, & \text{ otherwise.}
\end{cases}
\end{equation}
\item The solution to the $(\mu,\rho_{\phi})$ trade-off analysis (eq. \ref{spec.tradeoff}) when $\beta$ is restricted to $[0,1]$ is given by
\begin{equation}\label{brand.4}
\beta^{\star}=\begin{cases}
0,&\text{ if } \frac{\E[X_{T,\rho_{\phi}}-X_{0}]}{\rho_{\phi}(X_{T,\rho_{\phi}})-\rho_{\phi}(X_{0})}\le \frac{\lambda}{1-\lambda}\\
1, & \text{ otherwise.}
\end{cases}
\end{equation}
\end{enumerate}
\end{proposition}

\begin{remark}
As for the mean-variance framework, the risk-adjusted return of the tangency portfolio also plays a major role in the definition of the optimal $\beta^{*}$ in the $(\mu,\rho_{\phi})$ framework. In this case, however, we have $\beta^{\star}\in \{0,+\infty\}$ when short-sales are allowed, and $\beta^{\star} \in \{0,1\}$ when short-sales are restricted. Moreover, differently from what happens in the mean-variance framework, the solutions to the mean-spectral problems do not vary continuously with respect to the risk aversion $\lambda$.
\end{remark}

\begin{definition}\label{def.frontier}
A portfolio $X_{\beta,\gamma}$ belongs to the ${(\mu,\rho_{\phi})}$-\defin{efficient frontier} if there is no portfolio $X_{\beta',\gamma'}$ with $\E[X_{\beta',\gamma'}]\ge \E[X_{\beta,\gamma}]$ and $\rho_{\phi}(X_{\beta',\gamma'})\le \rho_{\phi}(X_{\gamma})$, with at least one of the two inequalities being strict.
\end{definition}
The problems in Equations \eqref{spec.lim1-2} and \eqref{spec.tradeoff} induce the same $(\mu,\rho_{\phi})$-efficient frontier. As in the mean-variance framework, the $(\mu,\rho_{\phi})$-efficient frontier consists of the linear combinations of the risk-free asset and the tangency portfolio (if short-sales are not allowed, only convex combinations are considered). Therefore, the solutions given in \cref{brand.3} and \cref{brand.4} show that the set of optimal solutions does not coincide with the set of portfolios in the efficient frontier. Moreover, differently from what happens in the mean-variance framework, the problems \eqref{spec.lim1-2} and \eqref{spec.tradeoff} are not equivalent in the sense that there is no one-to-one correspondence between the parameters $\mu$ and $\lambda$ such that, once these parameters are chosen appropriately, they induce the same solution.

\section{Application: Comparative Risk Aversion}\label{risk.aversion}

The comonotonic additive risk measures of \Cref{kusuoka_com.1} are defined through the weights attributed to the surpluses and losses.
The possibility of explicitly studying these weighting functions makes comonotonic additive risk measures interesting candidates to represent preferences. In \cite{brandtner2015decision}, the authors study preferences represented through coherent comonotonic additive risk measures. These preferences on $\mathcal{X}$ are denoted by $\preceq$ and, for any $X,Y \in \mathcal{X}$, are defined as $Y\preceq X$ if and only if $\rho_{\phi}(X)\le \rho_{\phi}(Y)$.

In the Arrow-Pratt (AP) setting \citep{pratt1964risk,arrow1965aspects}, the risk aversion of two agents can be compared through their certainty equivalents. The certainty equivalent $c_{\phi}\colon \mathcal{X}\rightarrow \mathbb{R}$ associated with $\rho_{\phi}$ is defined implicitly as $\rho_{\phi}(c_{\phi}(X))=\rho_{\phi}(X)$ for $X \in \mathcal{X}$. By cash additivity and normalization of $\rho_{\phi}$ one can always find such a $c_\phi$, which is given by $c_{\phi}(X)=-\rho_{\phi}(X)$ for all $X \in \mathcal{X}$.
Following \cite{brandtner2015decision}, we say that an agent whose preferences are represented by $\rho_{\phi_{1}}$ is more \defin{AP risk-averse} than another agent with preferences represented by $\rho_{\phi_{2}}$ if $\rho_{\phi_{1}}(X)\ge \rho_{\phi_{2}}(X)$ for all $X \in \mathcal{X}$. Equivalently, $\rho_{\phi_{1}}$ is more AP risk-averse than $\rho_{\phi_{2}}$ if $c_{\phi_{1}}(X)\le c_{\phi_{2}}(X)$ for all $X \in \mathcal{X}$. The intuition for this last definition is that more risk-averse agents require less money in exchange for lotteries. In most of the relevant literature, the \defin{AP coefficient of risk aversion} for spectral preferences is defined as
\begin{equation}\label{ap.coef}
R_{\phi}(p)=-\frac{\phi'(p)}{\phi(p)}, \quad \forall p \in [0,1].
\end{equation}
In this regard, \cite{brandtner2015decision} proved that $R_{\phi_{1}}(p)\ge R_{\phi_{2}}(p)$ for all $p \in [0,1]$ implies $\phi_{1}$ is more AP risk-averse than $\phi_{2}$, that is, $\rho_{\phi_{1}}(X)\ge \rho_{\phi_{2}}(X)$ for all $X \in \mathcal{X}$. However, they also proved that the converse is not true, i.e., it is possible that the AP coefficient does not correctly reflect the relative risk aversion of two spectral preferences. Therefore, the usage of the AP coefficient of risk aversion in \cref{ap.coef}---which is a classical tool to order the risk aversion of different agents---is incompatible with spectral preferences, in the sense that the rank based on the AP coefficient does not necessarily match the rank based on the certainty equivalents.

Another inconsistency in comparative risk aversion for spectral preferences is that the AP risk aversion ordering between two agents is not necessarily the same as the risk aversion ordering based on the Ross (R) criterion \citep{ross1981some}. This criterion generalizes the framework of Arrow and Pratt by considering uncertain levels of wealth, which will be represented by a random variable $X \in \mathcal{X}$, by defining the incremental risk premium as the amount an agent is willing to pay to avoid changing her wealth from $X$ to $X+Y$, where $Y \in \mathcal{X}$. In the spectral framework of \cite{brandtner2015decision}, the \defin{incremental risk premium} induced by a spectral risk measure $\rho_{\phi}$ is defined as $R_{\phi}(X,Y)\coloneqq\rho_{\phi}(X+Y)-\rho_{\phi}(X)$, for $X,Y \in \mathcal{X}$ being two non-constant random variables satisfying $\E[Y|X]=0$. The hypothesis of zero conditional expectation is aligned with the interpretation of $Y$ as a random variable adding noise to $X$, without being correlated with it.
In this framework, an agent whose preferences are represented by a spectral risk measure $\phi$ is \defin{R risk-averse} if $R_{\phi}(X,Y)\ge 0$ for all $X,Y \in \mathcal{X}$ satisfying the previously mentioned conditions. Accordingly, an agent whose preferences are represented by $\rho_{\phi_{1}}$ is more R risk-averse than another agent with preferences represented by $\rho_{\phi_{2}}$ if $R_{\phi_{1}}(X,Y)\ge R_{\phi_{2}}(X,Y)$ for all $X,Y \in \mathcal{X}$ satisfying the previously mentioned conditions.

Notice the change in the criterion for risk aversion: in the AP framework, the criterion depends only on the levels of $\rho_{\phi}$, while in the R framework the criteria also involve the increment in the risk. Propositions 3.3 and 4.3 of \cite{brandtner2015decision} show that if an agent with preferences $\rho_{\phi_{1}}$ is more R risk-averse than another agent with preferences $\rho_{\phi_{2}}$, then the same holds for their relative AP risk aversion. However, they also showed that the converse is not true, meaning that the ranking of preferences based on the AP risk aversion might not coincide with the ranking based on the R risk aversion, for spectral preferences. As a practical consequence for the AVaR, this implies that $p_{1}<p_{2}$ does not imply that the agents with preferences $\AVaR_{p_{1}}$ is more R risk-averse than an agent with preferences $\AVaR_{p_{2}}$. Also, \cite{brandtner2015decision} showed that similar inconsistencies hold for the exponential and power families of spectral risk measures.

\section{Concluding Remarks}\label{conclusion.1}
There are several reasons for which a risk manager may choose to measure financial risk with a comonotonic additive risk measure. First, the property of comonotonic additivity may be desirable for the application at hand, let it be, for instance, internal or external risk management. Also, the Kusuoka, spectral, and Choquet representations of comonotonic additive risk measures allow the risk manager to specify, very explicitly, how each level of the potential losses affects the final risk measurement.

Recent research, however, unveiled the fact that comonotonic additive risk measures cannot satisfy some other properties which, in some contexts, may be of utmost importance. In this paper, we present a comprehensive literature review focused on the properties that are necessarily absent for main classes of comonotonic additive risk measures. In addition, we provide new proofs showing that comonotonic additive risk measures cannot be surplus invariant (as defined in \cite{koch2017}, \cite{he2018}, and \cite{gao2020surplus}), except for the maximum loss (in the coherent case) and the value at risk (in the non-coherent case).

We present these issues in self-contained, separate sections, where we motivate the application at hand and discuss the respective conflict with comonotonic additivity. Also, we provide an appendix presenting the elementary of comonotonic random variables and comonotonic additive risk measures. Therefore, in addition to this paper's potential to serve as a reference guide for the experienced reader, the main content of our paper is also accessible to the audience not familiar with the theory of risk measures.


\newpage

\appendix

\section*{Appendix}\label{s:appendix}

\section{Background}\label{background}

The basic components of our setup are an atomless probability space $(\Omega, \mathcal{F},\Prob)$ and the space of essentially bounded random variables $\mathcal{X}\coloneqq L^{\infty}(\Omega, \mathcal{F},\Prob)$, which will serve as the domain of the risk measures considered. Also, we consider $\mathbb{R}$ as the sub-space of $L^{\infty}(\Omega, \mathcal{F},\Prob)$ containing the $\Prob$-almost surely constant random variables.
The random variables $X \in \mathcal{X}$ represent the discounted net value of a financial position at the end of the trading period. Accordingly, $X(\omega)>0$ stands for a gain and $X(\omega)<0$ represents a loss. Inequalities (and equalities) of the type $X>0$ should be understood in the $\Prob$-almost sure sense, unless otherwise specified. The notation $X \sim F_{X}$ stands for $F_{X}(x)\equiv \Prob(X\le x)$ $\forall x \in \mathbb{R}$, and $X\eqdist{}Y$ means that $F_{X}=F_{Y}$ point-wise. We denote the expectation and variance of $X \in \mathcal{X}$ as $\E[X]=\int X\,\mathrm{d}\Prob$ and $\V[X]=\int(X-\E[X])^{2}\,\mathrm{d}\Prob$, respectively.
For $X \in \mathcal{X}$ and $p \in (0,1]$, the left $p$-quantile of $X$ is defined as $q_{X}(p)\coloneqq\inf\{x \in \mathbb{R}\colon F_{X}(x)\ge p\}$ and $q_{X}(0)=\essinf(X)$. Also, for $X \in \mathcal{X}$ and $p \in [0,1)$, the right $p$-quantile of $X$ is defined as $q_{X}^{+}(p)\coloneqq\inf\{x \in \mathbb{R}\colon F_{X}(x)> p\}$ and $q_{X}^{+}(1)=\esssup(X)$. Also, we denote $x^{+}=\max\{x,0\}$ and $x^{-}=\max\{-x,0\}$. The terms ``increasing" and ``decreasing" are employed in the weak sense.

\subsection{Characterization of Comonotonic Random Variables}\label{characterization1}

The concept of comonotonicity dates back at least to Theorem 236 in \cite{hardy1934inequalities}, where it appears under the name of ``similarly ordered functions". In more recent treatments, especially in the literature of risk measures and premium principles, the concept is usually defined as follows:
\begin{definition}\label{def.comon1}
Consider the following definitions for $\mathbf{X}=(X_{1},X_{2},\dots,X_{n}) \in \mathcal{X}^{n}$.
\begin{enumerate}[noitemsep,nosep]
\item The random vector $\mathbf{X}$ is \defin{comonotonic} if
\begin{equation}\label{ineq.def1}
(X_{i}(\omega)-X_{i}(\omega'))(X_{j}(\omega)-X_{j}(\omega'))\ge 0\quad \forall\; i,j \in \{1,2\dots,n\} \,\Prob\otimes \Prob\mhyphen a.s.
\end{equation}
In this case, we also say that the random variables $X_{1},X_{2},\dots, X_{n}$ are comonotonic.
\item With $n=2$, $\mathbf{X}$ is said to be \defin{counter-comonotonic} if $(X_{1},-X_{2})$ is comonotonic. In this case we also say that the random variables $X_{1}$ and $X_{2}$ are counter-comonotonic.
\end{enumerate}
\end{definition}
The distinctive feature of comonotonic random vectors is that if one of its components varies, the others do not vary in the opposite direction. This property has a clear financial meaning: comonotonic random variables do not hedge each other. On the contrary, counter-comonotonic random couples are such that, whenever one of its components varies, the other does not vary in the same direction. Roughly speaking, the property of comonotonicity (counter-comonotonicity, respectively) implies a non-negative (non-positive, respectively) dependence between the random variables. Additionally, notice that constant random variables are comonotonic and counter-comonotonic with each other and with every random variable. Also, if $X_{1}$ and $X_{2}$ are comonotonic, and $f,g\colon \mathbb{R}\rightarrow \mathbb{R}$ are both increasing or both decreasing, then $f(X_{1})$ and $g(X_{2})$ are comonotonic. On the other hand, if $f$ is increasing and $g$ is decreasing (or vice-versa), then $f(X_{1})$ and $g(X_{2})$ are counter-comonotonic. The following classical Theorem gives alternative characterizations of comonotonicity.
\begin{theorem}\label{com.caract1}
(\cite{ruschendorf2013mathematical} - Theorem 2.14; \cite{dhaene2020comonotonic} - Theorem 4) Consider a random vector $\mathbf{X}=(X_{1},X_{2},\dots,X_{n})\in \mathcal{X}^{n}$ with marginal distributions $(F_{X_{1}},F_{X_{2}},\dots,F_{X_{n}})$ and joint distribution $F$. The following statements are equivalent:
\begin{enumerate}[noitemsep,nosep]
\item\label{com.1.1} The random vector $\mathbf{X}$ is comonotonic.
\item\label{com.2.1}The random vectors in $\{(X_{i},X_{j})\colon i,j \in \{1,2,\dots,n\}\}$ are comonotonic.
\item\label{com.3.1} $F(x_{1},x_{2},\dots,x_{n})=\min\{F_{X_{i}}(x_{i})\colon i \in \{1,2,\dots,n\}\}, \; \forall (x_{1},x_{2},\dots,x_{n}) \in \mathbb{R}^{n}$.
\item\label{com.4.1} $F(x_{1},x_{2},\dots,x_{n})\ge \Tilde{F}(x_{1},x_{2},\dots,x_{n})$ whenever $\Tilde{F}$ is a joint distribution of a random vector whose marginals are given by $(F_{X_{1}},F_{X_{2}},\dots,F_{X_{n}})$.
\item\label{com.5.1} For $U\sim \text{\normalfont \sffamily Uniform}(0,1]$, we have
\begin{equation}\label{qcomonot.1}
\mathbf{X}\eqdist{}(q_{X_{1}}(U),q_{X_{2}}(U),\dots,q_{X_{n}}(U)).
\end{equation}
Moreover, if $\mathbf{X}$ is comonotonic and $X_{1}$ is continuously distributed, then there exist increasing functions $f_{2},\dots,f_{n}\colon \mathbb{R}\rightarrow \mathbb{R}$ such that
\begin{equation*}
(X_{1},X_{2},\dots,X_{n})=(X_{1},f_{2}(X_{1}),\dots,f_{n}(X_{1})).
\end{equation*}
\end{enumerate}
\end{theorem}

The second item of the Theorem shows us that the essence of comonotonicity is captured by pairs of random variables. In this light, for simplicity, we focus on comonotonicity for pairs of random variables.
\Cref{com.3.1} says that the dependence structure of comonotonic random vectors are captured by a copula $C\colon [0,1]^{n}\rightarrow [0,1]$ which is given by $C(u_{1},u_{2},\dots,u_{n})=\min\{u_{i}\colon i \in \{1,2,\dots,n\}\}$. It is valid to mention that the joint distribution of a random vector is usually harder to estimate than its marginal distributions. For comonotonic random vectors though, \cref{com.3.1} says that the joint distribution can be readily recovered from the marginals. \Cref{com.4.1} shows that once the marginal distributions of a random vector are fixed, the comonotonic structure of dependence leads to the highest probability of joint losses. Random vectors are written as in \cref{qcomonot.1} are remarkably useful in the theory of comonotonic random variables. Also, \cref{com.5.1} is often taken as the definition of comonotonicity (see \cite{ruschendorf2013mathematical}, for instance). Notice that \cref{com.5.1} says nothing about the marginal distribution of the vector $(q_{X_{1}}(U),\dots,q_{X_{n}}(U))$. Therefore, it is valid to mention that $q_{X_{i}}(U)\sim F_{X_{i}}$ for all $i \in \{1,2,\dots,n\}$ (see Lemma A.23 of \cite{follmer2016stochastic}).
This fact, taken with the equivalence between items \ref{com.5.1} and \ref{com.3.1}, implies that one can obtain customized comonotonic random vectors in the sense that, for any n-tuple of marginal distributions $(F_{X_{1}},F_{X_{2}},\dots,F_{X_{n}})$, any random vector written as in \cref{qcomonot.1} is comonotonic with marginals given by $(F_{X_{1}},F_{X_{2}},\dots,F_{X_{n}})$. The final part of the Theorem indicates that the components of comonotonic random vectors share the same source of variability.
\begin{definition}
Let $\mathbf{X}=(X_{1},X_{2},\dots,X_{n})\in\ \mathcal{X}^{n}$ be any random vector with quantile functions $(q_{X_{1}},\dots,q_{X_{n}})$ and let $U\sim \text{\normalfont \sffamily Uniform}(0,1]$. A \textbf{comonotonic counterpart} of $\mathbf{X}$, denoted by $\mathbf{X}^{c}=(X_{1}^{c},X_{2}^{c},\dots,X_{n}^{c})$, is any random vector written as in \cref{qcomonot.1}.
\end{definition}
Notice that \cref{com.5.1} of \Cref{com.caract1} implies that every random vector $\mathbf{X} \in \mathcal{X}$ admits a comonotonic counterpart. Additionally, \cref{com.3.1} of the Theorem specifies the joint distribution of any comonotonic counterpart.

Examples of comonotonic and counter-comonotonic random variables are abundant in finance and actuarial science. Here we give just two examples, and refer to \cite{kass2001modern} and
\cite{denuit2006actuarial} for comprehensive treatments.
\begin{example}
The payoff of derivative securities, in particular call (resp., put) options, forms a comonotonic (resp., counter-comonotonic) pair when combined with the price of the underlying asset. Consider, for instance, a European call option with underlying asset $X$ and strike price $K>0$. Its payoff at the expiration date is given by $(X-K)^{+}$, which is a increasing function of $X$. The opposite holds between the underlying asset $X$ and the cash-flow of those underwriting the call options or buying European put options. The payoff of these positions (at the expiration date) are decreasing functions of $X$ and are respectively given by $-(X-K)^{+}$ and $(K-X)^{+}$.
\end{example}
\begin{example}
In actuarial science, comonotonic random variables appear, for instance, as the layers of a given loss $X \in \mathcal{X}^{+}$. These are contracts where the policy holder must pay the insurer a franchise of $a>0$ in exchange for the insurer to face the loss $X$ up to a limit $h>a$ (see \cite{denuit2006actuarial}, for more details). The loss faced by the insurer is then
\begin{equation*}
X_{[a,h]}=
\begin{cases}
0,\quad &\mbox{if}\; X<a,\\
X-a, \quad &\mbox{if}\; a\le X \le h,\\
h-a, \quad &\mbox{if}\; h<X
\end{cases}
\end{equation*}

Notice that any two layers $X_{(a_{1},h_{1}]}$ and $X_{(a_{2},h_{2}]}$ are increasing functions of the same random variable $X$ and, therefore, are comonotonic. For more applications of comonotonic random variables in finance and actuarial science see, for instance, \cite{dhaene2002applications}, \cite{deelstra2011overview}, \cite{denuit2012convex}, \cite{cheung2014reducing}, \cite{chen2015optimization}, and \cite{dhaene2020comonotonic}.
\end{example}

\subsection{The Basics on Risk Measures}\label{basic.1.appendix}

Since the landmark work of \citet*{artzner1999}, the theory of risk measures grew in symbiosis with that of insurance premium principles \citep{cai2020risk,kass2001modern} and choice theory \citep{delbaen_osaka,tsanakas2003risk}. Below we follow the traditional heuristic of interpreting risk measures as tools that help determine regulatory capital requirements for financial institutions. For a given financial position $X\in\mathcal{X}$, the real number $\rho(X)$ represents the minimum amount of capital, in terms of $t=0$ numéraire, that the financial institution must prudently invest (in liquid and stable assets) to have a ``reasonable" buffer against potential losses from $X$.
\begin{remark}
The interpretation of $\rho(X)$ as a quantity expressed in the $t=0$ numéraire is in line with the convention we adopted that the random variables in $\mathcal{X}$ represents discounted payoffs. This section would follow unchanged if $\rho(X)$ and the random variables in $\mathcal{X}$ were expressed in the numéraire of the terminal date.
\end{remark}
The following axioms can be motivated with this application in mind.
\begin{definition}\label{def.rho.1}
We call any functional $\rho\colon\mathcal{X} \rightarrow \mathbb{R}$ a \defin{risk measure}. Also, we say that
\begin{enumerate}[noitemsep,nosep,series=axioms]
\item\label{mono.1} (Monotonicity) $\rho$ is \defin{monotone} if $\rho(X) \ge \rho(Y)$ for all $X,Y \in \mathcal{X}$ such that $X\le Y$.
\item\label{cash.2} (Cash Additivity) $\rho$ is \defin{cash additive} if $\rho(X-b)=\rho(X)+b$ for all $b \in \mathbb{R}$ and $X \in \mathcal{X}$.
\item\label{ph.3} (Positive Homogeneity) $\rho$ is \defin{positive homogeneous} if $\rho(\lambda X)=\lambda \rho(X)$ for all $\lambda \ge 0$ and $X \in \mathcal{X}$.
\item\label{sub.4} (Subadditivity) $\rho$ is \defin{subadditive} if $\rho(X+Y)\le \rho(X)+\rho(Y)$ for all $X,Y \in \mathcal{X}$.
\item\label{cx.5} (Convexity) $\rho$ is \defin{convex} if $\rho(\lambda X + (1-\lambda)Y) \le \lambda \rho(X)+(1-\lambda)\rho(Y)$ for all $\lambda \in [0,1]$ and $X,Y \in \mathcal{X}$.
\item\label{norm.6} (Normalization) $\rho$ is \defin{normalized} if $\rho(0)=0$.
\item\label{law.7} (Law invariance) $\rho$ is \defin{law invariant} if $\rho(X)=\rho(Y) \text{ whenever } F_{X}=F_{Y}$ pointwise.
\item\label{com.8.1} (Comonotonic additivity) $\rho$ is \defin{comonotonic additive} if $\rho(X+Y)=\rho(X)+\rho(Y)$ for all $X,Y \in \mathcal{X}$ such that $(X,Y)$ is a comonotonic random vector.
\end{enumerate}
\end{definition}
The axiom of monotonicity requires that if the payout of a financial position $X$ is always smaller than that of $Y$, then the risk of $X$---and therefore its regulatory capital---must be greater than that of $Y$. Notice that, as is prevalent in the literature, we are assuming a unitary discount factor. In this light, the property of cash additivity says that if the future (unknown) result $X$ is depleted by a known amount $b$, becoming $X-b$, then the same amount, $b \in \mathbb{R}$, must be added to the original regulatory capital to maintain the same level of risk. Notice that cash additivity yields $\rho(X+\rho(X))=0$, which in turn implies that no capital reserves need to be made after $\rho(X)$ has been added to $X$.
Risk measures satisfying monotonicity and cash additivity are called \defin{monetary}.

Positive homogeneity requires the risk to vary ``linearly" with respect to variations in the financial position's size. The axiom of subadditivity reflects the notion that a merge does not create extra risks. Risk measures satisfying axioms \ref{mono.1} to \ref{sub.4} are called \defin{coherent} and were first studied in \cite{artzner1999} for discrete probability spaces and in \cite{delbaen2002coherent} for general probability spaces. Although the class of coherent risk measures is one of the most widely employed in both theory and practice, some authors consider the axiom of positive homogeneity too restrictive. For instance, \cite{frittelli2002putting} and \cite{follmer2002convex} argue that scaling up a financial position may create extra liquidity risks, which are not accounted for if the risk measure is positive homogeneous. This consideration has motivated the study of a less restrictive class of risk measures, which is defined by axioms \ref{mono.1}, \ref{cash.2}, and \ref{cx.5} and is called the class of \defin{monetary convex} risk measures \citep{heath2000back,follmer2002convex,frittelli2002putting}. Similar to subadditivity, the axiom of convexity aims to reflect diversification benefits. The axiom of normalization is implied by positive homogeneity and is usually introduced to ease the notation. For normalized risk measures, convexity, subadditivity, and positive homogeneity are linked as each pair of these axioms implies the remaining one.

The axiom of law invariance was introduced for risk measures and non-expected utility theory in the seminal contributions of \cite{kusuoka} and \cite{yaari1987dual}, respectively. This axiom was embraced by most of the literature because it is necessary for empirical applications where one only observes a set of data points (say, $X_1(\omega),\dots, X_n(\omega))$ drawn from an unknown probability distribution.

The axiom of comonotonic additivity says that the risk of a comonotonic sum equals the sum of the individual risks. This axiom was introduced in decision theory by \cite{schmeidler1986integral} and \cite{yaari1987dual}, in premium principles in \cite{wang1996premium} and \cite{wang1998comonotonicity}, and for risk measures in \cite{kusuoka} and \cite{acerbi2002spectral}. Of course, many others have contributed to the development of the theory of comonotonic additive risk measures. For details on the theory and applications of comonotonic additive risk measures see \cite{dhaene2002applications,dhaene2002theory}. The rationale behind this axiom is based on the strong dependence structure between comonotonic random variables, as mentioned in the preceding section. Such degree of dependence forbids hedging between comonotonic pairs, in the sense that their variations never ``compensate" each other. Therefore, it is argued that if $X$ and $Y$ are comonotonic, then the risk of the position $X+Y$ should be equal to the sum of the risks of $X$ and $Y$.

\begin{example}\label{example.var}
The widely used \defin{value-at-risk} is a monetary, positive homogeneous, law invariant, comonotonic risk measure. The value-at-risk of $X\in \mathcal{X}$ at the significance level $p \in [0,1]$ is defined as
\begin{math}
\VaR_{p}(X)=q_{-X}(1-p)=-q_{X}^+(p).
\end{math}
It is worth to mention that for $p \in [0,1]$ and $X \in L^{\infty}(\Omega,\mathcal{F},\Prob)$ the equality $q_{-X}(1-p)=-q_{X}(p)$ holds for almost all $p$.
\end{example}

Despite being widely employed in practice, the use of VaR for the determination of regulatory capital has been extensively criticized for two main reasons: first, it does not account for the size of the position below the $p$-quantile, allowing for instance that $\VaR_{p}(X)=\VaR_{p}(Y)$ even if the tail of $X$ below its $p$-quantile is heavier than that of $Y$ below its respective $p$-quantile. The second main critique is that $\VaR_{p}$ does not satisfy the axiom of subadditivity and, therefore, the value-at-risk does not capture the financial intuition behind this axiom.

\begin{example}\label{example.avar}
The \defin{average value-at-risk} ($\mathbf{\AVaR}$) circumvents both drawbacks of the $\VaR$. First, the average value-at-risk is coherent and, second, it takes into account all level of losses below the significance level being used. The average value-at-risk of $X \in \mathcal{X}$ at the significance level $p \in (0,1]$, denoted as $\AVaR_{p}(X)$, is defined as
\begin{equation}
\operatorname{AVaR}_{p}(X)=\frac{1}{p}\int_{0}^{p}\VaR_{q}(X)\,\mathrm{d}q.
\end{equation}
For $p=0$ it is defined as $\AVaR_{0}(X)=-\essinf X$. For $p\in(0,1]$, the $\AVaR_{p}(X)$ is a type of average of the $p(100)\%$ smaller values $X$. In particular, $\int_0^1 \VaR_p(X)\, \mathrm{d}p = E[-X]$ holds even if $X$ is not continuously distributed. If $X$ is continuously distributed, $\AVaR_{p}$ can be expressed as the expectation of the loss, $-X$, conditional on $X$ being no greater than $q_{X}(p)$, that is,
\begin{equation*}
\operatorname{AVaR}_{p}(X)=E[-X|X\le q_{X}(p)], \quad \text{where}\quad q_{X}(p) \text{ is any } p\text{-quantile of } X.
\end{equation*}
Risk measures related to the average value-at-risk can be found, for instance, in \cite{acerbi2002coherence} and \cite{dhaene2006risk}.

The average value-at-risk has gained the acceptability of practitioners, being included in the Basel accord for banking regulation and in the Swiss Solvency Test for insurance companies \citep{bcbs2019feb,keller2004white}. The AVaR's mathematical properties were scrutinized in \cite{acerbi2002coherence}. Since then, the $\AVaR$ was studied comprehensively by several authors and was shown to have a solid economic foundation in \cite{wang2021axiomatic}. Also, in the next section we will see that AVaR plays a central role in the theory of comonotonic risk measures. In short, all comonotonic additive risk measures satisfying certain additional properties can be represented as a mixture of $\AVaR$s at different significance levels.
\end{example}
Regardless of the risk measure being employed, a fundamental task of the regulator is to decide between accepting or not the position of the financial institutions. Once a risk measure, say $\rho$, has been chosen, this task can be accomplished through the set
\begin{equation*}
\mathcal{A}_{\rho}=\{X \in \mathcal{X}\colon \rho(X)\le 0\},
\end{equation*}
which can be viewed as a gauge according to which the financial institutions' position are appraised: a financial institution with a position represented by $X$ is deemed acceptable if and only if $X \in \mathcal{A}_{\rho}$.

Sets used to define the theoretical acceptability of financial positions are called \defin{acceptance sets}. These sets are of fundamental importance and can be taken as the primal concept in the theory of risk measures. Similar to risk measures, acceptance sets have an “axiomatic menu” of their own:
\begin{definition}\label{def.A.1}
We call any non-empty set $\mathcal{A}\subsetneq \mathcal{X}$ an \defin{acceptance set}. Also, we say that
\begin{enumerate}[noitemsep,nosep,series=accept.axioms.1]
\item (Monotonicity) $\mathcal{A}$ is \defin{monotone} if $X\le Y$ and $X \in \mathcal{A}$, imply $Y\in \mathcal{A}$.
\item (Boundedness on constants) $\mathcal{A}$ is \defin{bounded on constants} if $\inf\{m \in \mathbb{R}\colon m \in \mathcal{A}\}>-\infty$.
\item (Convexity) $\mathcal{A}$ is \defin{convex} if $\lambda \mathcal{A} + (1-\lambda)\mathcal{A} \subseteq \mathcal{A}$ whenever $\lambda \in [0,1]$.
\item (Conicity) $\mathcal{A}$ is a \defin{cone} if $\lambda \mathcal{A}\subseteq \mathcal{A}$ for all $\lambda \ge 0$.
\item (Normalization) $\mathcal{A}$ is \defin{normalized} if $\inf\{m \in \mathbb{R}\colon m \in \mathcal{A}\}=0.$
\end{enumerate}
\end{definition}
The property of monotonicity is a basic requirement for acceptance sets: if the regulator accepts $X$ while $Y$ pays more than $X$ with probability one, then $Y$ should also be acceptable. The property of boundedness on constants says that there is a lower bound on the size of certain losses that are acceptable. In fact, for much of the theory, the stronger property of normalization holds, which means that no certain loss is deemed acceptable. Acceptance sets satisfying monotonicity and boundedness on constants are called \defin{monetary}. The property of convexity corresponds to the requirement that diversification does not increase the risk. Conicity implies that the acceptability of a position should never be affected by changes in its scale.

An acceptance set $\mathcal{A}$ induces a real-valued functional through $\rho_{\mathcal{A}}(X)=\inf\{m \in \mathbb{R}\colon X+m \in \mathcal{A}\}$. The next Proposition shows that there is a correspondence between the properties of acceptance sets and those of risk measures.
\begin{theorem}\label{prop_relations}
The following illustrates the relation between acceptance sets and risk measures.
\begin{enumerate}[noitemsep,nosep]
\item\label{item1.prop_relations}If $\rho$ is monotone, so is $\mathcal{A}_{\rho}$. Reciprocally, if $\mathcal{A}$ is monotone, so is $\rho_{\mathcal{A}}$.
\item\label{item2.prop_relations} If $\rho$ is monetary, then it is $1$-Lipschitz continuous w.r.t.\ $\Vert\cdot\Vert_\infty$. In this case, $\mathcal{A}_{\rho}$ is non-empty, monetary, and closed w.r.t.\ $\Vert\cdot\Vert_\infty$. Reciprocally, if $\mathcal{A}$ is monetary, then $\rho_{\mathcal{A}}$ is monetary and, as a consequence, Lipschitz continuous w.r.t.\ $\Vert\cdot\Vert_\infty$.
\item\label{item3.prop_relations}If $\rho$ is monetary, then $\rho_{\mathcal{A}_{\rho}}=\rho$. Reciprocally, if $\mathcal{A}$ is monetary then $\mathcal{A}_{\rho_{\mathcal{A}}}$ corresponds to the $\Vert\cdot\Vert_\infty$ closure of $\mathcal{A}$.
\item\label{item4.prop_relations} If $\rho$ is convex, then $\mathcal{A}_{\rho}$ is convex. Also, if $\mathcal{A}$ is convex and monetary, then $\rho_{\mathcal{A}}$ is convex and monetary.
\item\label{item5.prop_relations}If $\rho$ is positive homogeneous, then $\mathcal{A}_{\rho}$ is a cone. Reciprocally, if $\mathcal{A}$ is a cone, then $\rho_{\mathcal{A}}$ is positive homogeneous.
\item\label{item6.prop_relations} If $\rho$ is non-zero monotone and comonotonic additive, then it is positive homogeneous and Lipschitz continuous w.r.t.\ $\Vert\cdot\Vert_\infty$.
\item\label{item7.prop_relations} If $\rho$ is cash additive and normalized, then $\mathcal{A}_{\rho}$ is normalized. Reciprocally, if $\mathcal{A}$ is normalized, then $\rho_{\mathcal{A}}$ is normalized.
\end{enumerate}
\end{theorem}
\begin{proof}
For items \ref{item1.prop_relations}-\ref{item5.prop_relations} see \cite{follmer2016stochastic}. \Cref{item6.prop_relations} is proved in Proposition 2.5 of \cite{koch2018}.
To prove the first assertion of item \ref{item7.prop_relations}, notice that the cash additivity of $\rho$ implies $\inf\{m\in \mathbb{R}\colon \rho(m) \le 0\}=\inf\{m\in \mathbb{R}\colon \rho(0)\le m\}=\rho(0)=0$, where the last equality follows by the normalization hypothesis on $\rho$. Conversely, if $\mathcal{A}$ is normalized one has $0=\inf\{m \in \mathbb{R}\colon m \in \mathcal{A}\}=\rho_{\mathcal{A}}(0)$, which implies that $\rho_{\mathcal{A}}$ is normalized.
\end{proof}

\subsection{Representation Theorems}\label{rep}
A major theoretical appeal of comonotonic additive risk measures is that they can be represented as certain integrals. These representations are simpler than, for instance, those of coherent risk measures provided in \cite{artzner1999} and \cite{delbaen2002coherent}, and of monetary convex risk measure provided in \cite{frittelli2002putting} and \cite{follmer2002convex}. They are also important in clarifying the incompatibilities existing between comonotonic additivity and the desirable properties we discuss. The following definition is necessary to this section's main Theorem.

\begin{definition} Consider the following continuity properties:
\begin{enumerate}[noitemsep,nosep]
\item A risk measure $\rho\colon\mathcal{X}\rightarrow \mathbb{R}$ is \defin{continuous from above} if $\rho(X_{n})\rightarrow \rho(X)$ whenever $X_{n}\downarrow X\; \Prob\mhyphen$a.s.
\item A risk measure $\rho\colon\mathcal{X}\rightarrow \mathbb{R}$ is \defin{continuous from below} if $\rho(X_{n})\rightarrow \rho(X)$ whenever $X_{n}\uparrow X\; \Prob\mhyphen$a.s.
\end{enumerate}
\end{definition}
For the next definition, denote by $\mathcal{H}$ the set of increasing concave functions $h\colon [0,1]\rightarrow [0,1]$ satisfying $h(0)=0$ and $h(1)=1$. For the reasons put forward in remarks \ref{remark.expec}, \ref{remark.dominance}, and \ref{remark.relative}, we refer to the elements of $\mathcal{H}$ as \defin{concave distortions}. Also, for $h \in \mathcal{H}$, define the set function $c_{h}\colon \mathcal{F}\rightarrow [0,1]$ as $c_{h}(A)=h(\Prob(A))$ for all $A \in \mathcal{F}$.
\begin{definition}
The \defin{Choquet integral} of $X\in \mathcal{X}$ with respect to $h$ is given by
\begin{equation*}
\int X\,\mathrm{d}c_{h}=\int_{-\infty}^{0}(c_{h}(X>x)-1)\,\mathrm{d}x+\int_{0}^{+\infty}c_{h}(X>x)\,\mathrm{d}x.
\end{equation*}
\end{definition}
The following Theorem dates back to \cite{kusuoka} and \cite{acerbi2002spectral}. We denote by $\mathcal{M}([0,1])$ the set of probability measures on the Borel sets of $[0,1]$.
\begin{theorem}\label{kusuoka_com.1}
(\cite{follmer2016stochastic}) Consider a risk measure $\rho\colon\mathcal{X}\rightarrow \mathbb{R}$. Then the following are equivalent:
\begin{enumerate}[noitemsep,nosep]
\item\label{rep.teo.1} $\rho$ is coherent, comonotonic, law invariant, and continuous from above.
\item\label{rep.teo.2} $\rho$ has the following \defin{Kusuoka} representation
\begin{equation}\label{avarmix1.1}
\rho(X)=\int \operatorname{AVaR}_{t}(X)\mu(\mathrm{d}t),\; X \in \mathcal{X},\text{ for some } \mu \in \mathcal{M}([0,1]).
\end{equation}
\item\label{rep.teo.3} $\rho$ has the following \defin{Choquet} representation
\begin{equation}\label{choquet1.1}
\rho(X)=\int-X\,\mathrm{d}c_{h}, X\in \mathcal{X}, \text{ for some }h \in \mathcal{H}.
\end{equation}
\item\label{rep.teo.4} $\rho$ has the following \defin{spectral} representation
\begin{equation}\label{spectral.simple.1}
\rho(X)=h(0+)\AVaR_{0}(X)+\int_{0}^{1}q_{-X}(t)h'(1-t)\,\mathrm{d}t,X \in \mathcal{X}, \text{ for some }h \in \Psi.
\end{equation}
\end{enumerate}
\end{theorem}
\begin{proof}
All assertions are proved in \cite{follmer2016stochastic}. The equivalence between items \ref{rep.teo.1} and \ref{rep.teo.2} was proved in their Theorem 4.93. The equivalence between items \ref{rep.teo.2} and \ref{rep.teo.3} follows by their Corollary 4.77. The equivalence between items \ref{rep.teo.3} and \ref{rep.teo.4} was given in their Theorem 4.70.
\end{proof}
\begin{remark}\label{remark.one-to-one}
The equivalence between items \ref{rep.teo.2}, \ref{rep.teo.3}, and \ref{rep.teo.4} is possible because there exists a one-to-one correspondence between $\mathcal{M}([0,1])$ and $\mathcal{H}$ (see \cite{follmer2016stochastic}, \cite{acerbi2002spectral}, or \cite{dhaene2012remarks} for details).
\end{remark}
\begin{remark}
Since AVaR is coherent, continuous from above, and law invariant, any risk measure in the form given in \cref{avarmix1.1} has the same properties (see \cite{follmer2016stochastic} sec. 4.6, p.246). Also, the convexity of the risk measures in \cref{avarmix1.1} is a consequence of AVaR being convex and the fact that convex combinations of convex function are convex (see Proposition 2 in \cite{acerbi2002spectral}).
\end{remark}
Denote by $\Phi$ the set of non-negative decreasing functions $\phi\colon [0,1]\rightarrow \mathbb{R}_{+}$ such that $\int_{0}^{1}\phi(t)\,\mathrm{d}t=1$.
\begin{corollary}\label{spectral.brend}
A risk measure $\rho\colon\mathcal{X}\rightarrow \mathbb{R}$ is a coherent comonotonic additive law invariant continuous from above and continuous from below if and only if
\begin{equation}\label{spectral.simple.c}
\rho(X)=-\int_{0}^{1}q_{X}(t)\phi(t)\,\mathrm{d}t,X \in \mathcal{X}, \text{ for some }\phi \in \Phi.
\end{equation}
\end{corollary}
\begin{remark}
The above corollary gives us a simplified version of the spectral representation for continuous from below risk measures. It assumes particular importance in \Cref{port}.
\end{remark}
\begin{remark}
Notice that if $h$ is not concave in the Choquet representation, then the risk measure associated with it is not subadditive. In this case, the map $t \mapsto h'(1-t)$ is not decreasing and, as a consequence, the associated spectral representation will not be coherent.
\end{remark}
\begin{remark}\label{remark.expec}
Notice that risk measures represented as in \cref{choquet1.1} can be regarded as a generalization of the famous ``expectation formula" for $-\E[\cdot]$, namely
\begin{equation}\label{exp.formula.1}
\E[-X]=\int_{-\infty}^{0}(\Prob(-X\ge x)-1)\,\mathrm{d}x+\int_{0}^{+\infty}\Prob(-X>x)\,\mathrm{d}x.
\end{equation}
In fact, one can obtain $\rho(X)=\E[-X]$ through the representations given in \cref{choquet1.1} and \cref{spectral.simple.c} by taking $h(t)=t$.
\end{remark}
\begin{remark}\label{remark.dominance}
Notice that, for $\psi \in \mathcal{H}$, we have
\begin{align*}
h(\Prob(-X>x))&=h(\Prob(-X>x)+0(1-\Prob(-X>x)))\\
&\ge \Prob(-X>x)h(1)+(1-\Prob(-X>x))h(0)\\
&=\Prob(-X>x)\quad \forall \;x \in \mathbb{R}
\end{align*}
for all $x \in \mathbb{R}$. Therefore, one can compare \cref{exp.formula.1} and \cref{choquet1.1} to conclude that $\rho(X)\ge \E[-X]$ for all $X \in \mathcal{X}$ and $\rho$ as defined in \Cref{kusuoka_com.1}.
\end{remark}
\begin{remark}\label{remark.relative}
Most distortion functions used in practice are continuous at zero, which implies
\begin{equation*}
\lim_{x \rightarrow + \infty}h(\Prob(-X>x))=0.
\end{equation*}
In view of the representation given in \cref{choquet1.1} this means that, as the losses' size grows and $\Prob(-X>x)\rightarrow 0$, the distorted probability also goes to zero. Nonetheless, the relative distortion is greater for higher losses, in the sense that for $x_{1}<x_{2}$, the concavity of $h$ implies
\begin{equation*}
\frac{h(\Prob(-X>x_{2}))}{\Prob(-X>x_{2})}\ge \frac{h(\Prob(-X>x_{1}))}{\Prob(-X>x_{1})},
\end{equation*}
which captures the idea that high losses should be more penalized.
\end{remark}
For the following lemma, denote by $\mathcal{H}^{*}$ the set of increasing functions $h\colon [0,1]\rightarrow [0,1]$ satisfying $h(0)=0$ and $h(1)=1$.
\begin{proposition}\label{teo.kou}
(\cite{follmer2016stochastic} - Theorem 4.88; \cite{kou2016} - Lemma 1) A monetary risk measure $\rho\colon \mathcal{X}\rightarrow \mathbb{R}$ is comonotonic and law invariant
if and only if there exists $h \in \mathcal{H}^{*}$ such that
\begin{equation}
\rho(X)=\int(-X)\,\mathrm{d}c_{h}, \; \forall X \in \mathcal{X}.
\end{equation}
Also, the above integral can be written as
\begin{align}\label{new.integral}
\rho_h(X)= &0 \vee q_{-X}(1-t_i)+ 0 \wedge q_{-X}(1-t_s)\\
& + \int_{0 \wedge q_{-X}(1-t_i)}^{0\wedge q_{-X}(1-t_s)}(h\circ \Prob(-X>x)-1)dx + \int_{0 \vee q_{-X}(1-t_i)}^{0\vee q_{-X}(1-t_s)}h\circ \Prob(-X>x)dx,
\end{align}
where $t_s=\sup \{t \in [0,1):h(t)=0\}$, $t_i=\inf \{t \in (0,1]:h(t)=1\}$.
\end{proposition}
\begin{proof}
For any $X \in \mathcal{X}$ and $h \in \mathcal{H}^*$ it holds that (assuming $0<t_s\le t_i <1$):
\begin{equation}\label{aux.new.1}
h(\Prob(-X>x))\begin{cases}
=0, &\mbox{ if } 0 \le \Prob(-X>x) < t_s\\
\in [0,1] &\mbox{ if } t_s \le \Prob(-X>x) \le t_i\\
=1 &\mbox{ if } t_i < \Prob(-X>x) \le 1,
\end{cases}
\end{equation}
Making use of the equivalence $\Prob(-X \le x) \ge t \Leftrightarrow x \ge q_{-X}(t)$ for all $t \in (0,1)$, the previous inequalities can be equivalently expressed as
\begin{equation}\label{aux.new.2}
h(\Prob(-X>x))\begin{cases}
=0, \mbox{ if } x \ge q_{-X}(1-t_s)\\
\in [0,1], \mbox{ if } q_{-X}(1-t_i)\le x \le q_{-X}(1-t_s)\\
1, \mbox{ if } x < q_{-X}(1-t_i),
\end{cases}
\end{equation}
\Cref{aux.new.2} holds for the cases $t_s=0$ and/or $t_i=1$ but, in these cases, the inequalities $\Prob(-X>x)<t_s$ and $t_i<\Prob(-X>x)$ should be made weak in \cref{aux.new.1}. We will make use of \cref{aux.new.2} to conclude that the formula in \cref{new.integral} holds for all $X \in \mathcal{X}$ and $h \in \mathcal{H}^*$. Therefore, this proof also holds for the cases $t_s=0$ and/or $t_i=1$. 

There are three cases to consider (notice that $q_{-X}(1-t_i)\le q_{-X}(1-t_s)$):
\begin{mycases}
 \item $q_{-X}(1-t_s)\le 0$:
   In this case,
 \begin{align*}
 \rho_h(X)&=\int_{-\infty}^{q_{-X}(1-t_i)}0dx+\int_{q_{-X}(1-t_i)}^{q_{-X}(1-t_s)}[h\circ \Prob(-X>x)-1]dx+\int_{q_{-X}(1-t_s)}^{0}-1dx\\
 & =-q_{-X}(1-t_s)+\int_{q_{-X}(1-t_i)}^{q_{-X}(1-t_s)}[h\circ \Prob(-X>x)-1]dx.
 \end{align*}
 \item $q_{-X}(1-t_s)\le 0 \le q_{-X}(1-t_i)$: In this case we have
 \begin{align*}
 \rho_h(X)&=\int_{-\infty}^{q_{-X}(1-t_i)}0 dx + \int_{q_{-X}(1-t_i)}^{0}[h\circ \Prob(-X>x)-1]dx\\
 &+\int_0^{q_{-X}(1-t_s)}h\circ \Prob(-X>x)dx+\int_{q_{-X}(1-t_s)}^{\infty}0dx.
 \end{align*}
 \item $0 \le q_{-X}(1-t_i)$: In this case, it holds that
 \begin{align*}
 \rho_h(X)=\int_{0}^{q_{-X}(1-t_i)}h\circ \Prob(-X>x)dx + \int_{q_{-X}(1-t_i)}^{q_{-X}(1-t_s)}h\circ \Prob(-X>x)dx + \int_{q_{-X}(1-t_s)}^{\infty}0dx.
 \end{align*}
\end{mycases}
As it is clear to see, the risk measurements obtained in all cases corresponds to what would be provided by the formula in \cref{new.integral}.
\end{proof}

The class of risk measures defined above clearly contains the class put forward in \Cref{kusuoka_com.1} and \Cref{spectral.brend}. The larger class in \Cref{teo.kou} assumes particular importance in \Cref{elicit}. The above risk measures are not necessarily convex but are monotone. For the interested readers, non-monotone Choquet integrals were studied in \cite{wang2020characterization}.
\begin{definition}
\citep{acerbi2002spectral} A function $\delta\colon \mathbb{R}\rightarrow \mathbb{R}$ is called the \defin{Dirac delta function} if
\begin{equation}\label{dirac.1}
\int_{a}^{b}f(x)\delta(x-c)\,\mathrm{d}x=f(c),\; \forall c \in (a,b).
\end{equation}
The function $\delta'\colon \mathbb{R}\rightarrow \mathbb{R}$ is the first derivative of the Dirac delta function $\delta$ if
\begin{equation}\label{dirac.deriv.1}
\int_{a}^{b}f(x)\delta'(x-c)\,\mathrm{d}x=-f'(c), \;\forall c \in (a,b).
\end{equation}
\end{definition}
\begin{remark}
\Cref{dirac.1} is, in fact, an abuse of notation. Arguably, in the case of \cite{acerbi2002spectral}, it was used to avoid long detours from the innovative ideas that were being presented. We believe that, to focus on the main elements of the theory, it is reasonable to make use of the same definition.
\end{remark}
\begin{proposition}
The following are examples of risk measures in their Kusuoka, Choquet, and spectral representations:
\begin{enumerate}[noitemsep,nosep]
\item For $p \in [0,1]$ the risk measure $\VaR_{p}$ can be recovered from the Kusuoka representation by using $\mu(dt) =-t\delta'(t-p)dt$. From the Choquet representation by using $h(t)=1_{(t> p)}$, and from the spectral representations by using $\phi(t)=\delta(t-p)$.
\item For $p \in [0,1]$ the risk measure $\AVaR_{p}$ can be recovered from the Kusuoka representation by using
\begin{equation*}
\mu(A)=1_{p}(A)=
\begin{cases}
1 &\mbox{if}\quad p \in A,\\
0 &\mbox{otherwise}.
\end{cases}
\end{equation*}
From the Choquet representation by using $h(t)=(1/p)(\min\{t,p\})$, and from the spectral representation by using $\phi(t)=(1/p)1_{(t\le p)}$.
\item The risk of $X \in \mathcal{X}$ as measured by the $\operatorname{MinVaR}$ \citep{cherny2009new} is given by
\begin{equation*}
\operatorname{MinVaR}(X)=-\E[\min(X_{1},\dots,X_{n})]
\end{equation*}
where $\{X_{i}\}_{i=1}^{n}$ are $n \in \mathbb{N}$ independent copies of $X$. It can be recovered from the Kusuoka representation by using $\mu$ such that
\begin{equation*}
n(1-t)^{n-1}=\int_{(t,1]}s^{-1}\mu(\mathrm{d}s).
\end{equation*}
It can be recovered from the Choquet representation by using $h(t)=1-(1-t)^{n}$, and it can be recovered from the spectral representation by using $\phi(t)=n(1-t)^{n-1}$.
\end{enumerate}
\end{proposition}

\newpage

\bibliographystyle{apalike}

\bibliography{biblio}

\end{document}